%% file: BinaryRevize2012.tex
\documentclass{article}
\usepackage{amsmath,amssymb, amsthm}
\usepackage[all]{xy}
\usepackage{graphicx}
\usepackage{tikz}
\usetikzlibrary{arrows}
\usetikzlibrary{calc}

\input{obrazkyA.inc}

\newtheorem{thm}{Theorem}
\newtheorem{lem}[thm]{Lemma}
\newtheorem{cor}[thm]{Corollary}
\newtheorem{claim}{Claim}
\newtheorem*{lemnul}{Lemma}
\newtheorem*{thmnul}{Theorem}

\theoremstyle{definition}
\newtheorem{defn}[thm]{Definition}

\theoremstyle{remark}

\theoremstyle{remark}

\theoremstyle{remark}
\newtheorem{convention}[thm]{Convention}

\theoremstyle{remark}
\newtheorem{example}[thm]{Example}

\newcommand{\RP}{{\rm R}}
\newcommand{\normh}{h_{\M}}
\newcommand{\I}{\mathbb C}
\newcommand{\ii}{\mathbf c}
\newcommand{\iv}[1]{\overline{#1}\hspace{2pt}}

\newcommand{\abs}[1]{\vert#1\vert}

\newcommand{\p}{\mbox{\rm pref}}
\newcommand{\s}{\mbox{\rm suff}}
\newcommand{\kk}[1]{\underline{#1}} %mirror stuff
\newcommand{\cislo}[1]{\noindent {\marg #1}\hspace{3pt}}
\font\marg=cmbx8
\newcommand{\M}{\text{\tiny\bf m}}

\newcommand{\E}{{\rm Eq}}
\newcommand{\ee}{{\rm eq}}

\begin{document}

\title {Binary equality sets are generated by two words}

\author{\v St\v ep\'an Holub\\
\ \\
\small{Faculty of Mathematics and Physics, Charles University}\\
\small{186 75 Praha 8, Sokolovsk\'a 83, Czech Republic}\\
\small{\texttt{holub@karlin.mff.cuni.cz}}}
\date{}

\maketitle

\begin{abstract}
We show that the equality set $\E(g,h)$ of two non-periodic binary morphisms $g,h:A^*\to \Sigma^*$ is generated by at
most two words. If the rank of $\E(g,h)=\{\alpha,\beta\}^*$ is two, then $\alpha$ and $\beta$ start (and end) with
different letters.

This in particular implies that any binary language has a test set of cardinality at most two.
\end{abstract}

\newcounter{cond} %list of conditions
\newcounter{cla} %list of claims
\newcounter{defi} %list in definitions
\renewcommand{\thecond}{\roman{cond}} %list of conditions
\renewcommand{\thecla}{\Alph{cla}} %list of claims
\renewcommand{\thedefi}{\alph{defi}} %list in definitions

%#############################################################################
\section*{About this version}
%#############################################################################
This is a revised version of my paper published (with the same title) in Journal of Algebra 259 (2003), 1--42.

A nucleus of the paper was a part of my Ph.D. thesis supervised by Ale\v s Dr\'apal (\cite{myphd}\,). The proof was completed during the postdoctoral stay in Turku granted by Turku Centre for Computer Science (TUCS). I am grateful especially to Juhani Karhum\"aki for making that stay possible. When writing the paper I discussed the topic with Vesa Halava, Tero Harju, Juhani Karhum\"aki and Juha Kortelainen. 

After the publication, I received comments from Elena (Petre) Czeizler, Markku Laine and V\'aclav Fla\v ska. The present version was carefully read by Ji\v{r}\'{\i}  S\'ykora. I am indebted to all of them for their effort, their comments and suggestions. 

The most important difficulty discovered was Lemma \ref{hyper technicke lema}, which does not hold as it stays in the published text. The corrected formulation given in the present version is the one from an early draft of the paper. Before the publication, I decided to use a stronger claim, which is in fact never needed in the paper, and which, as it turned out, is fallacious. Elena pointed out some difficulties in the proof, and Markku found a counterexample, making it clear that the stronger claim cannot be rescued.

The present version was written in October 2007 and August 2012. The material is partly reorganized, terminology is revised and most proofs rewritten. I hope that the text is now substantially more readable than the journal version, which is in some places excessively complicated and discouraging. On the other hand, there are no new discoveries and the overall argument remains the same.  

%#############################################################################
\section{Introduction}
%#############################################################################
 Binary equality language, i.e., the set on which two binary morphisms agree, is the most simple non-trivial example of
an equality language, the notion of which was introduced in \cite{ArtoEquality}. Equality languages in general play an
important role in formal language theory. For a survey and bibliography see \cite[Section 5]{handbookmorphisms}.

In the binary case, the morphisms are defined on a monoid generated by two letters. It was for the first time
extensively studied by K. \v Cul\'{\i}k II and J. Karhum\"aki in \cite{CuKabinary}. There, the main claim of our work was conjectured,
viz. that a binary equality language is generated by at most two words as soon as at least one of the morphisms is
non-periodic (or, equivalently, injective). An important step towards the proof of the conjecture was made in
\cite{EKRequality} where the following partial characterization was obtained.

\begin{thm}\label{char} The equality set of two binary morphisms $g,h:A^*\to \Sigma^*$, where $A=\{a,b\}$, has the following structure:
\begin{list}{\rm(\Alph{cla})}{\usecounter{cla}}
\item \label{char1} If $g$ and $h$ are periodic, then either $\E(g,h)=\{\varepsilon\}$ or
$$\E(g,h)=\{\varepsilon\}\bigcup \{\alpha\in A^+ \  \vert\   \frac{\ \abs{\alpha}_a}{\ \abs{\alpha}_b}=k\}$$ for some $k\geq0$ or $k=\infty$.
\item \label{char2} If exactly one morphism is periodic, then $$\E(g,h)=\alpha^*$$ for some word $\alpha\in A^*$.
\item \label{char3} If both $g$ and $h$ are non-periodic, then either $$\E(g,h)=\{\alpha, \beta\}^*$$ for some words
$\alpha, \beta \in A^*$, or $$\E(g,h)=(\alpha\gamma^*\beta)^*$$ for some words $\alpha, \beta, \gamma \in A^+$.
\end{list}
\end{thm}

The question remained open whether the second possibility of case (\ref{char3}), contradicting the conjecture, can
actually occur. In the present paper we show that the answer is negative and, moreover, if $\alpha$ and $\beta$ are
both nonempty, they start (and end) with different letters. This is formulated in the following main theorem.

\begin{thm}\label{hlavni veta}
Let $g,h:A^*\to \Sigma^*$ be non-periodic binary morphisms.
Let $\alpha$ and $\beta$, with $\alpha\neq\beta$, be nonempty minimal elements of $\E(g,h)$. Then
$$\p_1(\alpha)\neq\p_1(\beta)\quad\text{and}\quad\s_1(\alpha)\neq\s_1(\beta)\,.$$
%\end{list}
\end{thm}

As a trivial consequence we have a solution of the original question.

\begin{thm}\label{hlavni veta B}
Equality language of two nonperiodic binary morphisms is generated by at most two words. 
\end{thm}

I am not aware of any way how to prove
Theorem \ref{hlavni veta B} not using Theorem \ref{hlavni veta}.
\medskip

\noindent {\it Remark.} Later, in \cite{unique}, it has been shown that the equality sets generated by two words have a precise form. Namely, the following theorem holds true. 
\begin{thmnul}
Let $g$ and $h$ be distinct nonperiodic binary morphisms such that $\E(g,h)$ is generated by two words. Then there is a positive integer $i$ such that
$$ \E(g,h)=\{a^ib,\,ba^i\}^*, $$
up to renaming of the letters.
\end{thmnul}
The proof is based on Theorem \ref{hlavni veta}.
\medskip

A closely related problem is the size of a test set for binary languages. Indeed, if two morphisms agree on a language, it must be a subset of their equality language. In \cite{EKRequality}, it
is shown that all binary languages have a three element test set. Our result allows to cut down this bound to two. Let
us remark that this improvement is not a simple consequence of the fact that the equality language is generated by two
words --- the difference in the first (or last) letter is a necessary ingredient.

%#############################################################################
\section{Preliminaries}
%#############################################################################
\label{preliminaries}
In this section, we fix our notation and recall some basic facts. For a reference and unproved claims see  \cite{handbookcombinatorics} or \cite{oldlot}.
If $\Sigma$ is an alphabet, then let $\Sigma^*$ be the free monoid,
 and $\Sigma^+$ the free semigroup generated by $\Sigma$. The empty word is denoted by $\varepsilon$. 
 Any subset of $\Sigma^*$ is called a {\em language}. Let $A$ denote the two-letter alphabet $\{a,b\}$.

The length of the word is denoted by $\abs{u}$, and $\abs u_x$ denotes the number of occurrences of the letter $x$ in
$u$. 
 A {\em
prefix} of $u$ is any word $v\in \Sigma^*$ such that there exists a word $v'\in\Sigma^*$ with $u=vv'$. The set of all
prefixes of $u$ is denoted by $\p(u)$. A prefix $v$ of $u$ is {\em proper} if $v\neq\varepsilon$ and $v\neq u$.
Similarly, {\em suffix} and {\em proper suffix} are defined. The set of all suffixes of $u$ is denoted by $\s(u)$. The
first (the last resp.) letter of a nonempty word $u$ is denoted by $\p_1(u)$ ($\s_1(u)$ resp.). A word $v$ is
called a {\em factor} of $u$ if there exist words $w,w'\in \Sigma^*$ such that $u=w\,v\,w'$.

If $v\in \p(u)$
or $u\in \p(v)$, then we say that $u$ and $v$ are \emph{prefix-comparable} (or simply \emph{comparable}). The maximal common prefix of words $u$ and
$v$ is denoted by $u \wedge v$. 
If $u$ and $v$ are words, then the maximal {\em $u$\/-prefix} of $v$ is the maximal prefix of $v$ that is also a prefix of $u^i$ for some $i$.
Analogously, we define the maximal $u$-suffix of $v$.
We say that two words are \emph{suffix-comparable} if one is a suffix of the other.

Positive powers $u^n$ of a word are defined as usual, with $u^0=\varepsilon$.
We shall sometimes use also negative powers and work with elements of the free group, in order to simplify notation. This should not cause any confusion.  For example, if $u$ and $v$ are comparable, then we shall write $u^{-1}v$ even if $v$ is a proper prefix of $u$. In such a case, when $u=vw$, we have $u^{-1}v=w^{-1}$.

We shall define regular languages by regular expressions in a standard way. In particular, the language $\{u^i\ |\ i\geq 1\}$ is denoted by $u^+$ and $u^*=u^+\cup \,\{\varepsilon\}$. We say that $v$ is a prefix (suffix, factor resp.) of $u^+$ if $v$ is a prefix (suffix, factor resp.) of $u^i$ for some $i\geq 1$.

A nonempty word $u$ is called {\em primitive} if and only if $u=v^n$ implies $u=v$. The {\em primitive root} of a nonempty word $u$ is the
(uniquely given) primitive word $r$ such that $u\in r^+$. Words $u$ and $v$ are called  {\em conjugate} if $u=ww'$
and $v=w'w$ for some words $w$ and $w'$.

If we speak about minimality or maximality of some element, the implicit ordering is the prefix one, i.e., $v\leq u$ if
and only if $v\in\p(u)$, and $v< u$ if moreover $v\neq u$. (While by the {\em shortest} word we mean the word with the smallest length.)

Let $u\in \Sigma^+$ be a word $u=l_1l_2\dots l_d$, with $d=\abs{u}$ and $l_i\in \Sigma$. Then the {\em reversal} of
the word $u$, denoted by $\iv u$, is obtained by inverting the order of the letters, viz. $$\iv u=l_dl_{d-1}\dots
l_1\,.$$ Let $g$ be an arbitrary morphism. The {\em reversal} of $g$ is the morphism denoted by $\iv g$, which has the same range and domain as $g$, and is defined
by
$$\iv g(x)=\iv{g(x)},$$ for each $x\in \Sigma$.
Note that in general $\iv g(u)$ does not equal to $g(\iv u)$ nor to $\iv{g(u)}$. Instead
$$\iv g(\iv u)=\iv{g(u)}. $$
All concepts and reasonings regarding prefixes are valid analogously for suffixes, reversals considered. We shall often
use the fact.

A morphism $g$ defined on $\Sigma$ is called {\em erasing} if $g(x)$ is empty for some $x\in \Sigma$. A morphism $g$ is {\em periodic} if there is a word $t$ such that $g(x)\in t^*$, for all words $x$ (or, equivalently, all letters $x$). Note that a binary morphism is periodic as soon as it is erasing. 

Let $S=T^+$ be a subsemigroup of $\Sigma^+$ generated by a set $T$.  The {\em rank} of $T$ is the cardinality of the
minimal set generating $S$. We can write
$$\mbox{\rm rank}(T)=\mbox{\rm rank}(S)=\mbox{\rm Card}(S \setminus S\cdot S)\,.$$
By the rank of a monoid $M$ we mean the rank of the semigroup $M\setminus \{\epsilon\}$.

 It is a well known fact that for
each set $M\subset \Sigma^+$ there exists the smallest free subsemigroup of $\Sigma^+$ containing $M$ and called its
{\em free hull}. A set generating a free semigroup is called a \emph{code}. If any two distinct elements of a code are neither prefix nor suffix comparable, the set is called a \emph{bifix code}. 

The {\em equality set} of two morphisms $g,h:\Delta^*\to \Sigma^*$ is defined by
$$\E(g,h)=\{u\in \Delta^*\ |\ g(u)=h(u)\}\,.$$
It is easy to verify that the set $\E(g,h)$ is a free submonoid of $\Delta^*$ generated by the set of its minimal
elements
$$\ee(g,h)=\E(g,h)\setminus (\E(g,h)\setminus \{\varepsilon\})^2\setminus\{\varepsilon\}\,.$$
Note that $\ee(g,h)$ is a bifix code.

 Let $g:A^*\to \Sigma^*$ be a nonperiodic binary morphism. By $z_g$ we denote the maximal common prefix of $g(ab)$ and
$g(ba)$, i.e.
$$ z_g=g(ab)\wedge g(ba)\,.$$
Since $g$ is nonperiodic, we have $\abs{z_g}<\abs{g(a)}+\abs{g(b)}$ by Lemma \ref{prim} below. If
$\p_1(g(a))\neq\p_1(g(b))$, i.e. $z_g=\varepsilon$, we say that $g$ is {\em marked}.

 Similarly we define $\kk z_g$ as the maximal common suffix of $g(ab)$ and $g(ba)$.
Note that
$$\kk z_g=\iv{\iv{g(ab)}\wedge \iv{g(ba)}}=\iv{z_{\iv g}} $$
and $\kk z_g=\varepsilon$ is equivalent to $\iv g$ being marked.

Cartesian product $\Delta^*\times \Delta^*$ is the set of ordered pairs $(u,v)$ of words. It can be seen as a monoid with
operation of catenation defined by $(u,v)(u',v')=(uu',vv')$, with the unit $(\varepsilon, \varepsilon)$. Such a monoid
is obviously not free, it is even not isomorphic to a submonoid of a free monoid.

 Let $g,h :\Delta^*\to \Sigma^*$ be two morphisms. The subset of $\Delta^*\times \Delta^*$ denoted by $\I(g,h)$ and defined by
 $$\I(g,h)=\{(u,v)\ |\ g(u)=h(v)\}$$
will be called the {\em coincidence set} of morphisms $g$ and $h$. It is generated by the set
$$\ii(g,h)=\I(g,h)\setminus (\I(g,h)\setminus \{(\varepsilon,\varepsilon)\})^2\setminus \{(\varepsilon,\varepsilon)\}\,.$$
Any pair $(u,v)\in \I(g,h)$ can be uniquely factorized into minimal pairs $(u_i,v_i)$ satisfying $g(u_i)=h(v_i)$. 
This is formulated in the following lemma.
 \begin{lem}\label{freely} Let $g$ and $h$ be non-erasing morphisms. Then $\I(g,h)$ is, as a submonoid of $\Delta^*\times \Delta^*$, freely generated by $\ii(g,h)$. Moreover, the set $\ii(g,h)$ is a bifix code.
\end{lem}

 Note that $(u,u)$ is an element of $\I(g,h)$ for each $u\in \E(g,h)$, and $\E(g,h)$ is given uniquely by $\I(g,h)$ as
 $$\E(g,h)=\{u \ |\ (u,u)\in \I(g,h)\}\,.$$

We present several combinatorial lemmas for future (often implicit) reference. Following three lemmas are part of the
folklore.

\begin{lem}\label{prim}
The words $u$ and $v$ commute if and only if they have the same pri\-mi\-tive root.
\end{lem}

\begin{lem}[Periodicity Lemma]
Let $u^+$ and $v^+$ have a common prefix of length $\abs{u}+\abs{v}$. Then the words $u$ and $v$ commute.
\end{lem}

We shall often use the following lemma. It is based on the well known fact that a primitive word $t$ cannot satisfy
equality $tt=utv$, with $u$ and $v$ nonempty.

\begin{lem}\label{factor}
\begin{list}{\rm(\Alph{cla})}{\usecounter{cla}}
\item \label{factor0} Let $ww=uwv$. Then $u$, $v$ and $w$ commute.
\item \label{factor3} Let $uw$ be a prefix of $w^+$. Then $u$ and $w$ commute.
\item \label{factor1} Let $sw$ be a factor of  $w^+$. Then $s$ is a suffix of $w^+$.
%\item \label{factor2} Let $wp$ be a factor of $w^+$. Then $p$ is a prefix of $w^+$.
\item \label{factor3a} Let $uw$ be a suffix of $w^+$ and let $w$ be a prefix of $uw$. Then $u$ and $w$ commute.
\item \label{factor4} Let $u_1$, $u_2$, $w$, $w' \in\Sigma^+$  be words such that $w'$ and $w$ are conjugate, $\abs{u_1}\leq\abs{u_2}$, and
the words $u_1w'$, $u_2w'$ are prefixes of $w^{+}$. Then $u_1$ is a suffix of $u_2$ and $u_2u_1^{-1}$ commutes with
$w$.
\end{list}
\end{lem}

One more lemma, which is easy to prove:

\begin{lem}\label{gprefix}
Let $g:A^*\to A^*$ be a marked morphism and let $u,v\in A^*$. Then 
$g(u\wedge v)=g(u)\wedge g(v)$.
\end{lem}

The following nice lemma is a key fact about binary morphisms.

\begin{lem}\label{prefix}
Let $X=\{x,y\}\subseteq \Sigma^+$ be a nonperiodic set (i.e. $xy\neq yx$). Let $u\in xX^*$, $v\in yX^*$ be words such
that $\abs{u}, \abs{v} \geq \abs{xy\wedge yx}$. Then $u\wedge v = xy\wedge yx$.
\end{lem}

The proof is not difficult (see \cite{handbookcombinatorics}, p. 348). The lemma immediately implies that for a nonperiodic binary morphism $h$ and an arbitrary word $u\in A^+$ long enough, the word $z_h$ is a prefix of $h(u)$ and the $(\abs{z_h}+1)$\/-th letter of $h(u)$ indicates the first letter of $u$.
For any $u,v\in A^*$ we have
\begin{align}
z_h=h(au)z_h \wedge h(bv)z_h. \label{vsudypritomnyprefix}
\end{align}
It is now easy to see that the morphism $\normh$ such that
\begin{align}
\normh(u)=z_h^{-1}h(u)z_h, \label{marked definition}
\end{align} $u\in A$, is well defined. Moreover, it is marked, and the equality (\ref{marked definition}) holds for any
$u\in A^*$. We shall call it the {\em marked version} of $h$.

\noindent{\bf N.B.} The case $g=h$ is trivial. Throughout the paper we shall implicitly suppose $g\neq h$.

%#############################################################################
\section{Principal morphisms}
%#############################################################################
\label{principal}

In this section we show that at least one of the morphisms $g$ and $h$ can be supposed to be marked. As we shall see, this will make our research more convenient. The goal is achieved by choosing a suitable target alphabet.

\begin{defn}
We say that an (unordered) pair of binary morphisms $g,h: A^*\to \Sigma^*$ is {\em principal} if the target alphabet
$\Sigma$ generates the free hull of the set $\{g(a),g(b),h(a),h(b)\}$.
\end{defn}

The previous definition reflects the use of the term ``principal morphism'' in literature (see for example \cite{oldlot}, p. 170).
The advantages of principal morphisms stem from the following important property.

 \begin{lem}
Let $X$ be a finite subset of $\Sigma^*$ and let $Y$ be the minimal generating set of the free hull of $X$. Then for each element $y\in
Y$ there is a word $x\in X$ such that $y$ is a prefix (suffix resp.) of $x$.
\end{lem}

For the proof see \cite{defaut}, Lemma 3.1. For our purpose, note the following immediate corollary.

\begin{cor}\label{free hull}
Let $X$ be a finite subset of \,$\Sigma^*$ such that $\Sigma$ is the base of the free hull of $X$. Then
$$\Sigma=\{\p_1(u)\ |\ u\in X\}=\{\s_1(u)\ |\ u\in X\}.$$
\end{cor}

It is quite intuitive that choosing the minimal generating set of the free hull as the target alphabet has no influence on the coincidence set of the morphisms. The following lemma is formulated for binary morphisms, but it can be trivially extended to any domain alphabet.

\begin{lem}\label{principalni abeceda}
Let $g_1$, $h_1$ be morphisms $A^*\to \Sigma^*$. Then there is a principal pair of morphisms $g$, $h$ such that
\begin{align*}
\I(g,h)&=\I(g_1,h_1).
\end{align*}
Moreover, if $g_1$  \textup{(}$h_1$, $\iv{g_1}$, $\iv{h_1}$ resp.\textup{)} is marked, then such is also $g$ \textup{(}$h$, $\iv g$, $\iv h$ resp.\textup{)}.
\end{lem}

\begin{proof} Let $F\subset \Sigma^*$ be the free hull o the set $\{g_1(a),g_1(b),h_1(a),h_1(b)\}$ and let $C$ be an alphabet
whose cardinality equals the rank of $F$. Then $C^*$ and $F$ are isomorphic since they are both free monoids of the same rank; let $\varphi: C^*\to F$ be an isomorphism. Define morphisms $g, h: A^*\to C^*$
by
\begin{align*}
	g&=\varphi^{-1}\circ g_1, & h&=\varphi^{-1}\circ h_1.
\end{align*}
\[
\begin{tikzpicture}[>=angle 45, shorten >=0.5mm, shorten <=1mm]
\node (A) at (0,0) {$A^*$};
\node (F) at (2.5,1) {$F$};
\node (C) at (2.5,-1) {$C^*$};
\draw[->] (A)--(F) node [above, sloped,midway]{{\small $g_1,h_1$}};
\draw[->] (A)--(C) node [above,sloped, midway]{{\small $g,h$}};
\draw[->, transform canvas={xshift=-0.9ex}] (C)--(F) node [left, midway]{{\small$\varphi$}};
\draw[->, transform canvas={xshift=0.9ex}] (F)--(C) node [right, midway]{{\small$\varphi^{-1}$}};
\end{tikzpicture}
\]
Then $(g,h)$ is a principal pair of  morphisms, the above diagram commutes, and
$\I(g,h)=\I(g_1,h_1)$. The rest is obvious.
\end{proof}

The previous lemma shows that we can always, without loss of generality, suppose that the pair we work with is principal. 
We can now prove that this brings about markedness of one of the morphisms.

\begin{lem}\label{principalni markovanost}
Let $g$, $h$ be nonperiodic principal morphisms, with  $\ee(g,h)$ non\-empty. Then at least one of the morphisms $g$, $h$ is marked, and at least one of the morphisms $\iv g$, $\iv h$ is marked.

\end{lem}
\begin{proof}
Suppose that none of the morphisms is marked, therefore 
\begin{align*}
\p_1(g(a))&=\p_1(g(b)),& \p_1(h(a))=\p_1(h(b)).
\end{align*}
Let $x$ be a first letter of a word  $u\in\E(g,h)$. Then
\begin{align*}
\p_1(g(x))&=\p_1(h(x)),
\end{align*}
and Corollary \ref{free hull} implies that the morphisms are periodic, a contradiction.

Obviously, the morphisms $\iv g$, $\iv h$ are also principal, since the concept of the free hull is preserved under the reversal symmetry. This concludes the proof.
\end{proof}

%#############################################################################
\section{The block structure of the coincidence set}
%#############################################################################
\label{blocks}

In this section, we study the structure of the equality set of nonperiodic morphisms and their relation to the coincidence set. The previous section justifies why we shall always suppose that $g$ is marked.

Let  $u,v\in \Sigma^*$ be words such that $g(u)$ and $h(v)$ are comparable. Then the word $h(v)^{-1}g(v)$ is called an \emph{overflow} (the overflow may be a ``negative'' word if $g(v)$ is a prefix of $h(v)$).
Following lemmas show that the possibility to lengthen the
words $u$, $v$ to words $u'$, $v'$ such that $g(u')=h(v')$ is very restricted. Namely, the overflow $z_h$ is the only one admitting two different continuations.
\begin{lem} \label{critical}
Let $g$ and $h$ be binary morphisms, and let $g$ be marked. Let  $u,v\in A^*$ be words such that $g(u)$ and $h(v)$ are comparable and let
 $$ g(u)\neq h(v)z_h\,.$$
 Let $u_1,u_2,v_1,v_2\in A^+$ be words such that
\begin{align*}
g(uu_1)&=h(vv_1), & g(uu_2)&=h(vv_2).
\end{align*}
Then 
\begin{itemize}
	\item $\p_1(u_1)=\p_1(u_2)$, if $|g(u)|-|h(v)|<|z_h|$;
	\item $\p_1(v_1)=\p_1(v_2)$, if $|g(u)|-|h(v)|>|z_h|$.
\end{itemize}
\end{lem}
\begin{proof} If $u_1$, $u_2$, $v_1$ and $v_2$ satisfy the conditions of the lemma, then the same conditions are satisfied also by the
words $u_1uu_1$, $u_2uu_2$, $v_1vv_1$ and $v_2vv_2$ resp. Hence we can suppose that each of the words $u_1,u_2,v_1,v_2$
is longer than $z_h$.
 Consider three cases.
 
%\begin{itemize}
\cislo{1.} First suppose that $\abs{g(u)}<\abs{h(v)}+\abs{z_h}$. By (\ref{vsudypritomnyprefix}), $h(v)z_h$ is a
prefix of both $h(vv_1)$ and $h(vv_2)$ and
\begin{align*}
\p_1(g(u_1))=\p_1(g(u_2))=\p_1(g(u)^{-1}h(v)z_h)=x.
\end{align*}
Since $g$ is a marked morphism, this implies that $\p_1(u_1)=\p_1(u_2)$.
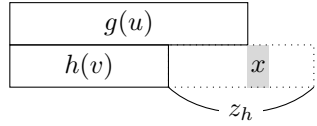
\[
\begin{tikzpicture}
 \horA{0}{90}{${g(u)}$}
 \dolA{0}{60}{${h(v)}$}
 \dolseda{90}{98}{$x$}
 \doldashA{60}{115}{}
 \doludolbra{60}{115}{$z_h$}
\end{tikzpicture}
\] 
\cislo{2.} Suppose on the other hand that $\abs{g(u)}>\abs{h(v)}+\abs{z_h}$. Then $h(v_1)$, $h(v_2)$ have the common prefix longer
than $z_h$ and $\p_1(v_1)=\p_1(v_2)$ is determined by the letter $x=\p_1((h(v)z_h)^{-1}g(u))$.
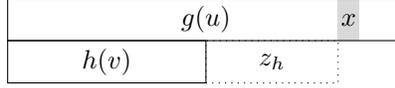
\[
\begin{tikzpicture}
 \horseda{125}{133}{$x$}
 \horA{0}{150}{${g(u)}$}
 \dolA{0}{75}{${h(v)}$}
 \doldashA{75}{125}{$z_h$}
\end{tikzpicture}
\] 
\cislo{3.} If $\abs{g(u)}=\abs{h(v)}+\abs{z_h}$, then, clearly, $g(u)=h(v)z_h$.
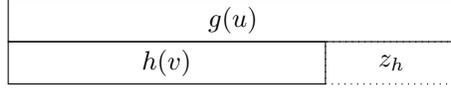
\[
\begin{tikzpicture}
 \horA{0}{170}{${g(u)}$}
 \dolA{0}{120}{${h(v)}$}
 \doldashA{120}{170}{$z_h$}
\end{tikzpicture}
\] 
\end{proof}

Previous lemma yields the following property.

\begin{lem}\label{cor1}
Let $g$ and $h$ be binary morphisms, and let $g$ be marked. Let $(c,d)$ and $(c',d')$ be distinct elements of
$\I(g,h)$, and suppose that $c$ and $c'$ are not comparable. Put
\begin{align*}
u&=c\wedge c',& v&=d\wedge d'.
\end{align*}
Then $$g(u)=h(v)z_h.$$
\end{lem}
\begin{proof}
We have $c=uu_1$ and $c'=uu_2$ where $u_1,u_2\in A^+$ and $\p_1(u_1)\neq \p_1(u_2)$. 

If $d$ and $d'$ are not comparable, then $d=vv_1$ and $d'=vv_2$ with $v_1,v_2\in A^+$ and $\p_1(v_1)\neq \p_1(v_2)$, and the claim follows from Lemma \ref{critical}.
 
If $d$ and $d'$ are comparable, then $\abs{g(u)}-\abs{h(v)}<0\leq \abs{z_h}$. Since
\begin{align*}
g(uu_1c)&=h(vv_1d), & g(uu_2c')&=h(vv_2d')
\end{align*}
with $u_1c, u_2c', v_1d, v_1d'\in A^+$, Lemma \ref{critical} yields a contradiction with $\p_1(u_1)\neq \p_1(u_2)$.
\end{proof}

\begin{example}
The previous corollary does not hold without the condition that $c$ and $c'$ are not comparable. Consider morphisms
\begin{align*}
g(a)&=a, & g(b)&=b,\\
h(a)&=a, & h(b)&=aab. 
\end{align*}
Then $(c,d)=(a,a)$, $(c',d')=(aab,b)$, $z_h=aa$, and
 $$g(c\wedge c')=g(a)=a\neq aa= h(\varepsilon)z_h=h(d\wedge d')z_h.$$
\end{example}

The ground for the characterization of the coincidence set is the following lemma.

\begin{lem}\label{ef}
Let $g$ and $h$ be binary morphisms, and let $g$ be marked. Let the words $e,f\in A^+$ satisfy following conditions:
\begin{list}{\rm(\roman{cond})}{\usecounter{cond}}
 \item \label{ef1} $z_hg(e)=h(f)z_h$
 \item \label{ef2} The words $e$, $f$ are minimal, i.e.: If $u\leq e$, $v\leq f$  and $z_hg(u)=h(v)z_h$, then either $u=v=\varepsilon$ or $u=e$ and $v=f$.
\end{list}
Then, given the first letter of $e$ or the first letter of $f$, the words $e$ and $f$ are determined uniquely.
\end{lem}
\begin{proof} Suppose $e$, $f$ and $e'$, $f'$ satisfy (\ref{ef1}) and (\ref{ef2}), and $\p_1(e)=\p_1(e')$. Put $c=e\wedge e'$, $d=f\wedge
f'$. Since $g$ is a marked morphism, we have
\begin{align}\label{pr1}
z_hg(e)\wedge z_hg(e')=z_hg(c)
\end{align}
by Lemma \ref{gprefix}.
From (\ref{vsudypritomnyprefix}) we deduce
\begin{align}\label{pr2}
h(f)z_h\wedge h(f')z_h=h(d)z_h.
\end{align}
Since $z_hg(e)=h(f)z_h$ and $z_hg(e')=h(f')z_h$, the equalities (\ref{pr1}), (\ref{pr2}) yield
\begin{align*}
z_hg(c)=h(d)z_h.
\end{align*}
Since $c$ is nonempty, we deduce from (\ref{ef2})  that $c=e=e'$ and $d=f=f'$. Similarly if $\p_1(f)=\p_1(f')$.
\end{proof}

This implies the following lemma.

\begin{lem}\label{coref}
Let $g$ and $h$ be binary morphisms, and let $g$ be marked.
\begin{list}{\rm(\Alph{cla})}{\usecounter{cla}}
\item The rank of $\I(g,\normh)$ is at most two.
\item If the rank of $\I(g,\normh)$ is two and $\ii(g,\normh)=\{(e,f),(e',f')\}$, then
\begin{align*}
\p_1(e)&\neq\p_1(e')\\
\p_1(f)&\neq\p_1(f').
\end{align*}
\end{list}
\end{lem}
\begin{proof} Recall that $\normh(u)=z_h^{-1}h(u)z_h$ to see that
$$\I(g,\normh)=\{(u,v)\in A^*\times A^*\ \vert\ z_hg(u)=h(v)z_h\}\,.$$
The rest is a consequence of Lemma \ref{ef}.
\end{proof}

Note that $(e,f)\in\ii(g,\normh)$ is just another formulation of the fact that $e$, $f$ are minimal words satisfying $z_hg(e)=h(f)z_h$, which are exactly conditions of Lemma \ref{ef}. The pairs $(e,f)$ and $(e',f')$ are often called \emph {blocks} of $g$ and $h$.

 The question on the structure of the equality set $\E(g,h)$ can be seen as a special case of the above
considerations. If conditions
\begin{align*}
u&=v,  & u_1&=v_1, & u_2&=v_2, &c&=d,  &c'&=d', &e&=f, &e'&=f',
\end{align*} are added, then we get the following modifications of  Lemma
\ref{critical}, Lemma \ref{cor1},  Lemma \ref{ef} and Lemma \ref{coref} with analogous proofs, which we omit.

\begin{lem} \label{critical revisited}
Let $g$ and $h$ be binary morphisms, and let $g$ be marked. Let  $u\in A^*$ be a word such that $g(u)$ and $h(u)$ are comparable, and
 $$ g(u)\neq h(u)z_h\,.$$
 Let $u_1,u_2\in A^+$ be words such that
\begin{align*}
g(uu_1)&=h(uu_1),\\ g(uu_2)&=h(uu_2).
\end{align*}
Then $\p_1(u_1)=\p_1(u_2)$.
\end{lem}

\begin{lem}\label{cor1 revisited}
Let $g$ and $h$ be binary morphisms, and let $g$ be marked. Let $c$ and $c'$ be incomparable elements of $\E(g,h)$. Put
 $u=c\wedge c'$.
Then $$g(u)=h(u)z_h.$$
\end{lem}

\begin{lem}\label{ef revisited}
Let $g$ and $h$ be binary morphisms, and let $g$ be marked. Let the word $e\in A^+$ satisfy following conditions:
\begin{list}{\rm(\roman{cond})}{\usecounter{cond}}
\item $z_hg(e)=h(e)z_h$
\item The word $e$ is minimal, i.e.: If $e_1$ is a prefix of $e$ and $z_hg(e_1)= h(e_1)z_h$, then $e_1=\varepsilon$ or $e_1=e$.
\end{list}
Then the word $e$ is determined uniquely by its first letter.
\end{lem}

\begin{lem}\label{coref revisited}
Let $g$ and $h$ be binary morphisms, and let $g$ be marked.
\begin{list}{\rm(\Alph{cla})}{\usecounter{cla}}
\item The rank of $\E(g,\normh)$ is at most two.
\item If the rank of $\E(g,\normh)$ is two and $\ee(g,\normh)=\{e,e'\}$, then
\begin{align*}
\p_1(e)&\neq\p_1(e').
\end{align*}
\end{list}
\end{lem}

Note that the previous lemma proves Theorem \ref{hlavni veta} for morphisms, which are marked from both sides. In the rest of the paper we show that this is essentially the only situation in which the equality set can have  rank greater than one.

Marked morphisms are in general much easier to deal with. That's why it is convenient to work with principal pairs, where one of the morphisms, say $g$, is marked. Moreover, it is always possible to use the marked version $\normh$ instead of $h$ to get a marked pair, and thus a better insight into the coincidence set of $g$ and $h$.

The block structure of the coincidence set of marked morphisms leads to an important concept of \emph{successor  morphisms} introduced first in \cite{EKRPCP}. Consider marked morphisms $g$ and $h$ such that $\ii(g,h)$ consists of two blocks $(e,f)$ and $(e',f')$. Let $w$ be an element of $\E(g,h)$. The equality $g(w)=h(w)$ can be uniquely split into a sequence of blocks. This means that $w$ is an element of $\{e,e'\}^+$, and in the same time an element of $\{f,f'\}^+$. It is now natural to define the successor morphisms $(g_1,h_1)$ by
\begin{equation}\label{jdeto} \begin{cases}  g_1(a)=e,\\
                   g_1(b)=e',
    \end{cases}
    \begin{cases}  h_1(a)=f,\\
                    h_1(b)=f'\,,
    \end{cases}
\end{equation}
and to formulate the previous considerations by the following lemma.
\begin{lem}\label{seq1}
 Let $g, h$ be marked morphisms such that 
 $$\ii(g,h)=\{(e, f),\, (e', f')\}.$$
Then the morphisms $g_1$, $h_1$ defined by \eqref{jdeto} are marked. 
If $w\in \E(g,h)$, then there is a unique word $w_1\in \E(g_1,h_1)$ such that 
\begin{align*}
g_1(w_1)&=h_1(w_1)=w.
\end{align*}
\end{lem}
\begin{proof}
The morphisms $g_1$ and $h_1$ are marked by Lemma \ref{coref}. The existence and uniqueness of the word $w_1$ follows from $(w,w)\in \I(g,h)$, and from Lemma \ref{freely}.
\end{proof}

%#############################################################################
\section{The counterexample and its structure}
%#############################################################################
\label{counterexample section}

We now have all necessary ingredients for the proof of our main claim, Theorem \ref{hlavni veta}.
The course of the prove will be essentially by contradiction. We
shall assume that there exists a counterexample to the claim, and gradually show that such an assumption is contradictory. 

We first formulate what is understood as a counterexample.

\begin{defn}\label{counterexample}
We say that a pair of morphisms $(g,h)$ is a {\em counterexample} if
\begin{list}{\rm(\alph{defi})}{\usecounter{defi}}
 \item \label{defa} The rank of $\E(g,h)$ is at least two; 
 \item \label{defb} $g$ is marked and $h$ is not marked;
 \item \label{delky} $\abs{g(a)}>\abs{h(a)}$ and $\abs{g(b)}<\abs{h(b)}$.
\end{list}
\end{defn}

The third condition takes advantage of the symmetry of letters $a$ and $b$. Note that the strict inequalities do not harm generality, since $\abs{g(a)}=\abs{h(a)}$ or $\abs{g(b)}=\abs{h(b)}$ would imply $g=h$.
Since the letters $a$ and $b$ are not interchangeable anymore, we shall sometimes need the morphism $\pi$ defined by $\pi(a)=b$ and $\pi(b)=a$.

The following lemma yields basic information about the structure of the equality set of a  counterexample.

\begin{lem}\label{counterexample structure}
Let $(g,h)$ be a  counterexample. Then there exist nonempty words $\sigma$, $\nu_a$ and
$\nu_b$ such that $|\sigma|_a\geq 1$,
\begin{align*}
\p_1(\nu_a)&=a, & \p_1(\nu_b)&=b,
\end{align*}
the words $\sigma\nu_a$, $\sigma\nu_b$ are two distinct elements of $\ee(g,h)$ and
\begin{align}
 g(\sigma)&=h(\sigma)z_h, \label{zacatek}\\
 z_hg(\nu_a)&=h(\nu_a),\label{prostredeka}\\
 z_hg(\nu_b)&=h(\nu_b)\label{prostredekb}.
\end{align}
\[
 \begin{tikzpicture}
  \preghA{10}
  \ramec{0}{130} 
  \stredA{70}{40}
  \uzel{20}{$\sigma$}
  \uzel{100}{${\nu_l}$}
  \nadpisA{40}{70}{$z_h$} 
\end{tikzpicture}
\]
\end{lem}
\begin{proof} 
Let $u$ and $v$ be two distinct elements of $\ee(g,h)$. Note that $u$ and $v$ are not comparable, and put $\sigma=u\wedge v$, $u_1=\sigma^{-1}u$ and $v_1=\sigma^{-1}$. Clearly, $\p_1(u_1)\neq\p_1(v_1)$ and the choice of $\nu_a$ and $\nu_b$ is now obvious.
The equalities (\ref{zacatek}), (\ref{prostredeka}) and
(\ref{prostredekb}) are yielded by Lemma \ref{cor1 revisited}, and $|\sigma|_a\geq 1$ follows from $|g(b)|<|h(b)|$. 
\end{proof}

The equalities \eqref{zacatek}, \eqref{prostredeka} and \eqref{prostredekb} are of a special importance in the proof. They represent two points, where the structure of a  counterexample is well defined, and which therefore yield information for a combinatorial analysis.

The following lemma makes sure that the counterexample defined above deserves its name.
\begin{lem}\label{uprava}
Let $g_1$ and $h_1$ be nonperiodic binary morphisms such that $\ee(g_1,h_1)$ contains two elements $\alpha$ and $\beta$ with the same first letter. Then there is a counterexample $(g,h)$ such that $\E(g,h)=\E(g_1,h_1)$. 

Moreover, if $\iv{g_1}$ ($\iv{h_1}$ resp.) is marked, then also $\iv g$ ($\iv h$ resp.) is marked.
\end{lem}
\begin{proof}
 Lemma \ref{principalni abeceda} yields principal morphisms $g$ and $h$ such that $\E(g,h)=\E(g_1,h_1)$. By Lemma \ref{principalni markovanost} and by the symmetry of $g$ and $h$, we can suppose that $g$ is marked. Similarly, by the symmetry of $a$ and $b$, we can suppose that the condition \eqref{delky} of Definition \ref{counterexample} is satisfied. In order to see that $(g,h)$ is a counterexample, it remains to show that $h$ is not marked. If $h$ is marked, then both morphisms are marked, and $\p_1(\alpha)\neq \p_1(\beta)$ by Lemma \ref{coref revisited}, contrary to the assumption.

Markedness of reversals is conserved by Lemma \ref{principalni abeceda}.
\end{proof}

The further strategy is to show that there is no counterexample. We shall divide the investigation into several stages.

%#############################################################################
\section{When $z_h$ commutes} \label{zh commutes}
%#############################################################################
In this section we investigate two special situations, in which $z_h$ commutes with one of the image words. We show that those situations lead to a contradiction.
We start with a technical lemma, which will be the core of the proof. In the original version of this paper the claim had the following strong form:

\begin{lemnul}\label{hyper technicke lema stare}
{\em Let $g,h:A^*\to A^*$ be two marked morphisms. Let $u$, $u'$, $v$ and $v'\in A^*$ be words, and $s$, $r$, $q$ positive
integers such that
\begin{align*}
 g(a^sbu)&=h(a^sbu'), &
 g(a^{r}bv)&=h(a^{q}bv').
\end{align*}
Then $s=r=q$.}
\end{lemnul}

However, as Markku Laine pointed out by constructing an example, this claim does not hold. The example is as follows.
\begin{example}
Let 
\begin{align*}
g(a)&=a^2b^2, & h(a)&=a, \\
g(b)&=b, & h(b)&=b^2.
\end{align*}
Then
$$g(a^2b^2)=h(a^2ba^2bb)=a^2b^2a^2b^4,$$
and 
$$g(ab^2)=h(a^2b^2)=a^2b^4.$$
\end{example}

We therefore present a bit weaker version, which fits the purpose of this paper.

\begin{lem}\label{hyper technicke lema}
Let $g,h:A^*\to A^*$ be two marked morphisms. Let $u$ and $v\in A^*$ be words, and $s$, $r$, $q$ positive
integers such that
\begin{align}
 g(a^sbu)&=h(a^sbu),\label{nekk} \\
 g(a^{r}bv)&=h(a^{q}bv).
\end{align}
Then $s=r=q$.
\end{lem}

\begin{proof} Recall that we suppose $g\neq h$. (Obviously, only $r=q$ is forced if $g=h$.)
Let $g$ and $h$ be morphisms satisfying assumptions, and suppose that $s=r=q$ does not hold. Assume, moreover, that $g$ and $h$ are chosen such that the length of
$a^sbu$ is the smallest possible. We show that $a^sbu$ can be shortened, and hence obtain a contradiction.

We first prove that $g(a)$ and $h(a)$ do not commute. Suppose for a while that $\abs{g(a)}>\abs{h(a)}$, and that $t$ is the common primitive root of $g(a)$ and $h(a)$. From \eqref{nekk}, we deduce that $h(b)$ is comparable with $h(a^s)^{-1}g(a^s)$, which is an element of $t^+$. That is a contradiction with $h$ being marked.
Similarly if $\abs{g(a)}<\abs{h(a)}$. (Clearly, $g(a)=h(a)$ implies $g=h$.)
  
Let us continue the proof of the lemma. Lemma \ref{gprefix} applied once to $g$ and once to $h$ gives
\begin{align}\label{weq}
 g(a^{s}bu\wedge a^{r}bv) &=g(a^{s}bu)\wedge g(a^{r}bv)=h(a^{s}bu)\wedge h(a^{q}bv)=h(a^{s}bu\wedge a^{q}bv).
\end{align}
%\begin{itemize}
\cislo{1.} If $s\neq r$ and $s\neq q$, then  (\ref{weq}) yields $$g(a^i)=h(a^j),$$ with $i=\min(s,r)$, $j=\min(s,q)$.
Therefore the words $g(a)$ and $h(a)$ commute, a contradiction.

\cislo{2.} Suppose next, by symmetry, $s=r$ and $s\neq
q$. Put $m=\min(s,q)$. Equality (\ref{weq}) implies
\begin{align}\label{1712}
g(a^sbw)=h(a^m),
\end{align}
where $w=u\wedge v$.

 The set $\I(g,h)$ contains elements $(a^sbu,a^sbu)$ and $(a^sbw,a^m)$, whence the rank of $\I(g,h)$ is two.
 Let $(e,f)$ and $(e',f')$ be the blocks of $g$ and $h$, and let $g_1$, $h_1$ be their successor morphisms defined by
\eqref{jdeto}. 

By symmetry, suppose that $\p_1(f)=a$.
Equality (\ref{1712}) implies that there is a positive integer $p$ such that $f=a^{p}.$ Since $g(a)$ and $h(a)$ do not commute, we deduce that $e\notin
a^+$ and thus $\abs{e}>s$. Since $a^sbu$ and $a^qbv$ are elements of $\{f,f'\}^*$, both $s$ and $q$ are multiples of
$p$. Put
\begin{align*}
s_1&=\frac s p,&q_1&=\frac q p,
\end{align*}
and define words $u_1$ and $v_1$ by
  \begin{align*}
g_{1}(u_{1})&=a^sbu, & h_{1}(u_{1})&=a^sbu, \\
g_{1}(v_{1})&=a^{s}bv, & h_{1}(v_{1})&=a^{q}bv.
  \end{align*}
Since $f=a^p$, the words $u_1$ and $v_1$ can be factorized as
\begin{align*}
u_1&=a^{s_1}bu_2, &  v_1&=a^{q_1}bv_2,
\end{align*}
with $u_2,v_2\in A^*$.
Therefore
\begin{align*}
g_{1}(a^{s_1}bu_2)&=h_{1}(a^{s_1}bu_2)=a^sbu,\\
g_1(a^{q_1}bv_2)&=h_1(a^{s_1}bv_2)=a^sbv.
\end{align*}
Inequality $s\neq q$ implies $s_1\neq q_1$, and $\abs{e}>s$ yields $\abs{a^{s_1}bu_2}<\abs{a^sbu}$. This completes the
proof.
\end{proof}

The following two claims exploit the previous lemma. The words $\sigma$, $\nu_a$ and $\nu_b$ are as in Lemma \ref{counterexample structure}.

\begin{claim}\label{nekomutuje s b}
There is no counterexample such that $z_h$ commutes with $g(b)$ and $\p_1(\sigma)=b$.
\end{claim}
\begin{proof}
Suppose that $(g,h)$ is such a counterexample, and let $t$ be the common primitive root of $z_h$ and $g(b)$.
Let $b^{\ell}$ be the maximal $b$-prefix of $\sigma\nu_a$ and $b^k$ be the maximal $b$-prefix of $\nu_b\sigma$.  
 Then $g(b)^\ell$ is the maximal $t$\/-prefix of $g(\sigma\nu_a)$ and $z_hg(b)^k$ is the maximal $t$\/-prefix of   $z_hg(\nu_b\sigma)$.

Suppose that $h(b)$ commutes with $g(b)$. Since $\abs{h(b)}>\abs{g(b)}$, the equality $g(\sigma\nu_a)=h(\sigma\nu_a)$ implies that $g(a)$ is comparable with $g(b)^{-\ell}h(b)^\ell$, a contradiction with $g$ being marked. Therefore $h(b)$ and $g(b)$ do not commute. 

By Lemma \ref{factor}\eqref{factor3}, the maximal $t$\/-prefix of $h(b)z_h$ is shorter than $\abs{h(b)t}$. This implies, by \eqref{vsudypritomnyprefix}, that all words $h(bu)$ long enough have the same maximal $t$\/-prefix. In particular, the maximal $t$-prefix of $h(\sigma\nu_a)$ is the same as the maximal $t$-prefix of $h(\nu_b\sigma)$.
From $h(\sigma\nu_a)=g(\sigma\nu_a)$ and $h(\nu_b\sigma)z_h=z_hg(\nu_b\sigma)$ we deduce that
$g(b)^\ell=z_hg(b)^k$ and $k\neq \ell$.

Put $\sigma'=b^{-\ell}\sigma$ and note that $\sigma'$ is nonempty since $\sigma$ contains the letter $a$. Then
\begin{align*}
z_hg(b^k\sigma')&=h(b^\ell\sigma')z_h,& z_hg(\nu_b\sigma)&=h(\nu_b\sigma)z_h,
\end{align*}
and Lemma \ref{hyper technicke lema}, applied to morphisms $\normh\circ\pi$ and $g\circ\pi$, yields a
contradiction.
\end{proof}

\begin{claim}\label{nekomutuje s a}
There is no counterexample such that $\p_1(\sigma)=a$, $z_h$ commutes with $h(a)$, and the common primitive root of $z_h$ and $h(a)$ is a suffix of $g(a)$.
\end{claim}
\begin{proof}
As in the previous proof, suppose that $(g,h)$ satisfies assumptions of the claim and let $t$ be the common primitive root of $z_h$ and $h(a)$.
 Let $a^{\,\ell}$ be the maximal $a$-prefix of $\sigma\nu_b$, and $a^k$ be the maximal $a$-prefix of $\nu_a\sigma$. 

First, suppose that $g(a)$ commutes with $h(a)$ and $z_h$. Since $g$ is marked,  
the maximal $t$-prefix of $z_hg(\nu_a\sigma)$ is $z_hg(a)^k$. From  \eqref{vsudypritomnyprefix}, we deduce that the word $z_h$ is the maximal $t$\/-prefix of $h(bu)z_h$ for any $u$. Hence the maximal $t$\/-prefix of $h(\nu_a\sigma)z_h$ is $h(a)^kz_h$. The equality $z_hg(\nu_a\sigma)=h(\nu_a\sigma)z_h$ now yields $z_hg(a)^k=h(a)^kz_h$, a contradiction with $\abs{g(a)}>\abs{h(a)}$. Therefore $g(a)$ and $h(a)$ do not commute.

Since, by assumption, $t$ is a suffix of $g(a)$, Lemma \ref{factor}\eqref{factor3} implies that the maximal $t$-prefix of $g(\sigma\nu_b)=h(\sigma\nu_b)$ is equal to the maximal $t$-prefix of $g(a)$. Using  \eqref{vsudypritomnyprefix} as above, we deduce from that this maximal $t$-prefix is equal to  $h(a)^\ell z_h$. In this case, the equality $z_hg(\nu_a\sigma)=h(\nu_a\sigma)z_h$ implies $z_hh(a)^\ell z_h=h(a)^kz_h$ whence
$z_h=h(a)^{k-\ell}$.

For $\sigma'=a^{-\ell}\sigma$ we obtain
\begin{align*}
z_hg(a^\ell\sigma')&=h(a^{k}\sigma')z_h, & z_hg(\nu_a\sigma)&=h(\nu_a\sigma)z_h.
\end{align*}
Since $g(a)$ and $h(a)$ do not commute, we deduce that $\sigma \notin a^+$, whence $\p_1(\sigma')=b$ and
morphisms $\normh$, $g$ satisfy assumptions of Lemma  \ref{hyper technicke lema}, a contradiction.
\end{proof}

%#############################################################################
\section{Case: $\iv g$ is not marked}
%#############################################################################
\label{g not marked}

In this section we deal with the situation when $\iv g$ is not marked. Note that then $\iv h$is marked by Lemma \ref{principalni markovanost}, and verify that $(\iv h\circ \pi,\iv g\circ \pi)$ is also a counterexample. Recall that $\pi$ exchanges letters $a$ and $b$, and it is applied in order to satisfy the condition \eqref{delky} of Definition \ref{counterexample}.
This allows to suppose 
 \begin{align}
 \abs{\kk z_g}\geq\abs{z_h} \label{delsi zetko}.
 \end{align}
More precisely, if $\abs{\kk z_g}<\abs{z_h}$, then we consider $(\iv h\circ \pi,\iv g\circ \pi)$, instead of $(g,h)$.

The equality (\ref{vsudypritomnyprefix}) applied to reversals implies that $\kk z_g$ is a suffix of any $g(u)$ long enough.
Especially,
\begin{align}
&\text{$\kk z_g$ is a suffix of $g(a)^+$}, \label{vsudypritomnysuffixA}\\
&\text{$\kk z_g$ is a suffix of $g(b)^+$}. \label{vsudypritomnysuffixB}
\end{align}
Since $z_h$ is a suffix of $g(\sigma)$, which is suffix comparable with $\kk z_g$, we deduce from \eqref{delsi zetko} that 
 \begin{align} \label{zhsufzg}
  z_h\in\s(\kk z_g).
 \end{align}
\[
\begin{tikzpicture}
\preghA{10}
\jensigma{80}{55}
\hordashA{35}{80}{{\scriptsize $\kk z_g$}}
\uzel{25}{$\sigma$}
\podpisA{55}{80}{$z_h$}
\end{tikzpicture}
\]
The following claim excludes the situation of this section.

\begin{claim}\label{spor zg non-empty}
There is no counterexample with $\iv g$ not marked.
\end{claim}
\begin{proof}
Suppose that $(g,h)$ is such a  counterexample.

\cislo{1.} Suppose first $\p_1(\sigma)=a$.
The equality $g(\sigma)=h(\sigma)z_h$ yields $h(a)\in \p(g(a))$, and $z_hg(\nu_a)=h(\nu_a)$ implies that $h(a)z_h$
is a prefix of $z_hg(a)$. Thus $z_hh(a)=h(a)z_h$. 

Let $t$ be the common primitive root of $h(a)$ and $z_h$. From \eqref{zhsufzg} we deduce that $t$ is a suffix of $\kk z_g$, and \eqref{vsudypritomnysuffixA} together with $\abs{g(a)}>\abs{h(a)}\geq\abs t$ yields that $t$ is a suffix of  $g(a)$.
This is a contradiction with Claim \ref{nekomutuje s a}.

\cislo{2.} Suppose then that $\p_1(\sigma)=b$.
From \eqref{vsudypritomnysuffixB} and (\ref{zhsufzg}) we deduce that $z_hg(b)$ is a suffix of $g(b)^+$.
Equalities $g(\sigma)=h(\sigma)z_h$ and  $z_hg(\nu_b)=h(\nu_b)$ imply that $g(b)$ is a prefix of $h(b)$, and that $h(b)$ is comparable with $z_hg(b)$ respectively. Therefore $g(b)$ is a prefix of $z_hg(b)$, and Lemma \ref{factor}\eqref{factor3a} yields that $g(b)$ and $z_h$ commute, a contradiction with Claim \ref{nekomutuje s b}.
\end{proof}

%#############################################################################
\section{Case: $\iv h$ is not marked}\label{h not marked}
%#############################################################################

In this subsection we consider the situation when $\iv g$ is marked and $\iv h$ is not. We shall not exclude this case directly. Instead we reduce it to the case when both $\iv g$ and $\iv h$ are marked.

To accomplish this plan we first we need a description of possible counterexample structure that is more precise than Lemma \ref{counterexample structure}.

\begin{lem}\label{presna struktura}
Let $(g,h)$ be a  counterexample. Then one of the following possibilities takes place.
\begin{list}{\rm(\Alph{cla})}{\usecounter{cla}}
\item \label{presna struktura1} There exist words $\sigma, \mu_a,
\mu_b\in A^+$, and $\tau\in A^*$ such that
\begin{align*}
\ee(g,h)&=\{\sigma\mu_a\tau,\,\sigma\mu_b\tau\},
\end{align*}
where
\begin{align*}
z_hg(\mu_a)\kk z_h&=h(\mu_a),& g(\sigma)&=h(\sigma)z_h,\\
z_hg(\mu_b)\kk z_h&=h(\mu_b),& g(\tau)&=\kk z_h h(\tau),\\
\end{align*}
and
\begin{align*}
\p_1(\mu_a)&=a, & \p_1(\mu_b)&=b, & \s_1(\mu_a)&\neq\s_1(\mu_b)\,.
\end{align*}
\[
\begin{tikzpicture}
\preghA{10}
\ramec{0}{170}
\stredA{70}{40} \stredA{130}{150}
\uzel{20}{$\sigma$} \uzel{100}{$\mu_x$} \uzel{160}{$\tau$}
\podpisA{40}{70}{$z_h$} \podpisA{130}{150}{$\kk z_h$}
\end{tikzpicture}
\]
\item \label{presna struktura2} There exist words $\zeta, \mu, \rho, \eta\in A^+$ such that
\[
\ee(g,h)=\zeta(\rho\mu)^*\rho\eta=\zeta\rho(\mu\rho)^*\eta,
\]
and
\begin{align*}
g(\zeta)\kk z_h&=h(\zeta),& z_hg(\mu)\kk z_h&=h(\mu),& \p_1(\mu)&\neq\p_1(\eta), \\
g(\rho)&=\kk z_h h(\rho)z_h, &z_hg(\eta)&=h(\eta), & \s_1(\mu)&\neq\s_1(\zeta).
\end{align*}
\[
\begin{tikzpicture}
\preghA{15}
\ramec{-5}{245}
\stredA{20}{40} \stredA{90}{60} \stredA{150}{170} \stredA{220}{190}
\uzel{8}{$\zeta$} \uzel{50}{$\rho$} \uzel{120}{$\mu$} \uzel{180}{$\rho$}  \uzel{233}{$\eta$}
\podpisA{20}{40}{$\kk z_h$} \podpisA{60}{90}{$z_h$} \podpisA{150}{170}{$\kk z_h$} \podpisA{190}{220}{$z_h$}
\end{tikzpicture}
\] 
\end{list}
\end{lem}
\begin{proof}
Let $\alpha$ and $\beta$ be two shortest elements of $\ee(g,h)$. Put $\sigma=\alpha \wedge \beta$, and similarly let $\tau$ be the longest common suffix of $\alpha$ and $\beta$. By Lemma \ref{cor1 revisited}, applied first to $g$ and $h$, and then to $\iv g$ and $\iv h$, we have 
\begin{align}
g(\sigma)&=h(\sigma)z_h, \label{alfa}\\
g(\tau)&=\kk z_hh(\tau)\,. 
\end{align}
Denote by $v_0$ and $v_1$ the words $\sigma^{-1}\alpha$ and $\sigma^{-1}\beta$. Clearly, $\p_1(v_0)\neq \p_1(v_1)$.

	\cislo{1.} First suppose that $v_0$ and $v_1$ are not suffix-comparable. Then with a suitable choice of $i,j\in \{0,1\}$ we have  $v_i=\mu_a\tau$, $v_j=\mu_b\tau$, and $\p_1(\mu_\ell)=\ell$ for both $\ell\in A$.
	
	Therefore $\{\sigma\mu_a\tau,\sigma\mu_b\tau\}=\{\alpha,\beta\}$.  We show that $\sigma$ is the unique prefix of $\alpha$ ($\beta$ resp.) satisfying \eqref{alfa}. 
	
	Suppose first that $\sigma_1\sigma_2=\sigma$ and $g(\sigma_1)=h(\sigma_1)z_h$. Then it is easy to see that also $\sigma_1v_i\in \E(g,h)$, $i=0,1$, a contradiction with $\alpha$ and $\beta$ being the shortest elements of $\ee(g,h)$.

Let then $v_i=w_1w_2$, for some $i\in \{0,1\}$, and  $g(\sigma w_1)=h(\sigma w_1)z_h$. Then $\sigma w_2$ is an element of $\E(g,h)$, which is shorter than $\sigma v_i$. Since $\alpha$ and $\beta$ are the shortest elements of $\ee(g,h)$, it remains that $\sigma w_2=\sigma v_{1-i}$. But then $v_0$ and $v_1$ are suffix-comparable, a contradiction.

	 We still have to show that the set $\{\alpha,\beta\}$ generates whole $\E(g,h)$. Suppose that $w$ is an element of $\E(g,h)$ such that neither $\alpha$, nor $\beta$ is a prefix of $w$, and consider words $w_i=w\wedge \sigma v_i$, $i=0,1$. Lemma \ref{cor1 revisited} implies that $g(w_i)=h(w_i)z_h$, for both $i=0,1$. 
It is easy to deduce that $w_0$ and $w_1$ cannot be both  equal to $\sigma$, a contradiction with the previous paragraph. Consequently, we have the case (A).

\cislo{2.}
	 Suppose now, by symmetry, that $v_1=uv_0$. Then $z_hg(u)=h(u)z_h$ and $\sigma u^* v_0$ is a subset of $\E(g,h)$. Moreover, $\sigma$ and $\sigma u$ are the only prefixes of $\sigma u v_0$ satisfying \eqref{alfa}. The proof is similar as above: any other prefix satisfying \eqref{alfa} allows to drop a part of the word, which contradicts the minimality of $\alpha$ and $\beta$. We omit details. This, in particular, implies that $u$ is not a suffix of $\sigma$.

We show that $\sigma u^* v_0$ generates the whole equality set. Suppose the contrary, and let $w$ be the shortest element of $\ee(g,h)$ that is not in $\sigma u^* v_0$. As above, the words $w_0=w \wedge \sigma v_0$ and $w_1=w\wedge \sigma uv_0$ satisfy $g(w_i)=h(w_i)z_h$, $i=0,1$. Therefore $w_0=\sigma$, by the previous paragraph. From $\p_1(u)\neq \p_1(v_0)$, one obtains that $w_1$ is strictly longer than $\sigma$, which implies $w_1=\sigma u$. Therefore $w=\sigma u w'$, for some $w'$. Hence $\sigma w'$ is an element of $\ee(g,h)$ shorter than $w$, and thus an element of $\sigma u^* v_0$. Therefore $w'\in u^*v_0$ and $w\in\sigma u^* v_0$, a contradiction.

We have seen that $u$ is not a suffix of $\sigma$. Also $\sigma$ cannot be a suffix of $u$, otherwise $\sigma u\sigma^{-1}\in \E(g,h)$ will contradict the minimality of $\alpha$ and $\beta$. We can therefore define $\rho$ as the longest common suffix of $u$ and $\sigma$. The word $\rho$ is not empty since $z_h$ is a suffix of both $g(u)$ and $g(\sigma)$, and $\iv g$ is marked.  Denote, $\eta=v_0$, $\zeta=\sigma\rho^{-1}$ and $\mu=u\rho^{-1}$.

Note that the word $\tau=\rho\eta$ is the longest common suffix of $\alpha$ and $\beta$. Lemma \ref{cor1 revisited} applied to $(\iv g, \iv h)$ yields $g(\rho\eta)=\kk z_h h(\rho\eta)$.	The verification of all claims in case (B) is now straightforward. 
\end{proof}

Note that the previous lemma proves, in particular, Theorem \ref{char}\eqref{char3}.

The following lemma allows to suppose that both $\iv g$ and $\iv h$ are marked, which was the task of this section. 

\begin{claim}\label{koncove zh prazdne}
Let $(g,h)$ be a counterexample. Then there exists also a counterexample $(g_1,h_1)$ such that both $\iv {g_1}$ and $\iv {h_1}$ are marked.
\end{claim}
\begin{proof} 
Suppose that $(g,h)$ is a  counterexample. Then $\iv g$ is marked by Claim \ref{spor zg non-empty}.
Suppose that $\kk z_{h}\neq \varepsilon$ and define $g_1$ and $h_1$ by
 \begin{align*}
 g_1(u)&=g(u), & h_1(u)&=\kk z_{h}h(u)(\kk z_{h})^{-1}.
 \end{align*}
 It is not difficult to see that the morphism $h_1$ is well defined, it is not marked while $\iv {h_1}$ is marked.
It remains to show that $\E(g_1,h_1)$ has rank at least two.
This is a consequence of the characterization presented in Lemma \ref{presna struktura}. (We shall use its notation.)

 \cislo{1.}If the case (\ref{presna
struktura1}) of Lemma \ref{presna struktura} takes place, then
 $$\{\tau\sigma\mu_a,\,\tau\sigma\mu_b\}\subset \E(g_1,h_1)\,.$$
\[
\begin{tikzpicture}
\ramec{0}{180}
\preghA{10}
\stredA{110}{60}
\uzel{30}{$\tau\sigma$} \uzel{145}{$\mu_x$}
\podpisA{60}{110}{$z_{h_1}=\kk z_{h} z_{h}$}
\end{tikzpicture}
\]
\cislo{2.} If, on the other hand, we have the case (\ref{presna
struktura2}) of the Lemma \ref{presna struktura}, then
 $$\{\rho\mu,\,\rho\eta\zeta\}\subset \E(g_1,h_1)\,.$$
\[
\begin{tikzpicture}
\ramec{0}{130}
\preghA{10}
\stredA{65}{15}
\uzel{87}{$\mu$} \uzel{8}{$\rho$}
\podpisA{15}{65}{$z_{h_1}=\kk z_{h} z_{h}$}
%%%%%%%%%%%%%%%%%%%%%
\ramec{170}{275}
\preghA{-160}
\stredA{235}{185}
\uzel{253}{$\eta\zeta$} \uzel{178}{$\rho$}
\podpisA{185}{235}{$z_{h_1}=\kk z_{h} z_{h}$}
\end{tikzpicture}
\]
Definitions in Lemma \ref{presna struktura} yield $\p_1(\mu_a)\neq \p_1(\mu_b)$ and $\p_1(\mu)\neq \p_1(\eta)$, whence the equality set has in both cases rank at least two.  
\end{proof}

%#############################################################################
\section{Case: $\iv g$ and $\iv h$ marked.} \label{both}
%#############################################################################

From now on we shall suppose that both $\iv g$ and $\iv h$ are marked. Consider Lemma \ref{presna struktura}. It is easy to note that the case \eqref{presna struktura1} of the lemma has to take place, and moreover, the word $\tau$ is empty. Therefore
  $$\ee(g,h)=\{\sigma \mu_a,\sigma \mu_b\},$$
with $\p_1(\mu_a)=a$, $\p_1(\mu_b)=b$, and $\s_1(\mu_a)\neq \s_1(\mu_b)$.

Note the following useful fact.
\begin{lem}\label{mirror}
	Let $(g,h)$ be a counterexample such that $\iv g$ and $\iv h$ are marked. Put $g_1=\iv g$ and $h_1=\iv{\normh}$. Then the pair $(g_1,h_1)$ is again a counterexample such that $\iv {g_1}$ and $\iv {h_1}$ are marked, and 
	$$\ee(g_1,h_1)=\{\iv \sigma \iv{\mu_a},\iv{\sigma}\iv{\mu_b}\}.$$
\end{lem}
\begin{proof}
The verification is straightforward.
\end{proof}

In this section, we will also need to assume that the pair $(g,h)$ is a \emph{shortest counterexample}. That is,  $|\sigma\mu_a|+|\sigma\mu_b|$ is as small as possible. Shortest counterexample have the following important properties, which can be summarized as: there are no repeated overflows. The proof is similar to the proof of Lemma \ref{hyper technicke lema}. If there is a repeated overflow, then we can decompose the counterexample into blocks, and find a shorter counterexample, namely the pair of successor morphisms. Since both $\iv g$ and $\iv h$ are marked, we will consider their blocks, which are easier to deal with.  
\begin{lem}\label{shortest1}
	Let $(g,h)$ be a counterexample such that $\iv g$ and $\iv h$ are marked. Let two nonempty prefixes $\sigma_1$ and $\sigma_2$ of $\sigma$ satisfy $g(\sigma_1)=h(\sigma_2)$. Then $(g,h)$ is not a shortest counterexample.
\end{lem}
\begin{proof}
 Lemma \ref{coref} applied to morphisms $\iv g$ and $\iv h$ implies that pairs $(\iv{\sigma\mu_a},\iv{\sigma\mu_a})$, $(\iv{\sigma\mu_b},\iv{\sigma\mu_b})$ and $(\iv{\sigma_1},\iv{\sigma_2})$ can be factorized into a sequence of pairs $(\iv e,\iv f)$ and $(\iv{e'},\iv{f'})$ such that $\iv g(\iv e)=\iv h(\iv f)$ and $\iv g(\iv{e'})=\iv h(\iv{f'})$. Turning to reversals and defining $g_1$ and $h_1$ as in \eqref{jdeto} we obtain words $w,w'\in \E(g_1,h_1)$ such that 
\begin{align*}
g_1(w)&=h_1(w)=\sigma\mu_a, & g_1(w')&=h_1(w')=\sigma\mu_b.
\end{align*}
  Note also that $\iv{g_1}$ and $\iv{h_1}$ are marked by Lemma \ref{seq1}.
 
 Since $(\sigma_1,\sigma_2)$ is a prefix of both $(\sigma\mu_a,\sigma\mu_a)$ and $(\sigma\mu_b,\sigma\mu_b)$, the words $w$ and $w'$ have a nonempty common prefix. From $g\neq h$, it is also easy to see that $|w|+|w|'<|\sigma\mu_a|+|\sigma\mu_b|$. Lemma \ref{uprava} concludes the proof. 
\end{proof}

\begin{lem}\label{shortest2}
	Let $(g,h)$ be a shortest counterexample such that $\iv g$ and $\iv h$ marked. Let prefixes $\sigma_1$, $\sigma_2$, $\sigma_1'$ and $\sigma_2'$ of $\sigma$ satisfy 
\begin{align}\label{16080}
	h(\sigma_2)^{-1}g(\sigma_1)=h(\sigma_2')^{-1}g(\sigma_1').
\end{align}
	 Then $\sigma_1=\sigma_1'$ and $\sigma_2=\sigma_2'$.
\end{lem}
Recall that we allow \eqref{16080} to be an equality of two ``negative'' words if $g(\sigma_1)<h(\sigma_2)$ and $g(\sigma_1')<h(\sigma_2')$.
\begin{proof}
Proceed by contradiction. Without loss of generality, we can suppose $|\sigma_1|>|\sigma_1'|$ and $|\sigma_2|>|\sigma_2'|$. 
Let $u_1$ be the longest common suffix of $\sigma_1$ and $\sigma_1'$, and let $u_2$ be the longest common suffix of $\sigma_2$ and $\sigma_2'$. We want to show that 
\begin{align}\label{16081}
	g(\sigma_1u_1^{-1})=h(\sigma_2u_2^{-1}). 
\end{align}
From \eqref{16080} and $g(\sigma\mu_a)=h(\sigma\mu_a)$, we deduce
\begin{align}\label{16082}
g(\sigma_1'\sigma_1^{-1}\sigma\mu_a)&=h(\sigma_2'\sigma_2^{-1}\sigma\mu_a).
\end{align}
If $u_1=\sigma_1'$ and $u_2=\sigma_2'$, then \eqref{16081} follows from \eqref{16082}.
Otherwise, we apply Lemma \ref{cor1} to morphisms $\iv g$ and $\iv h$ (note that the role of $\iv g$ and $\iv h$ is interchangeable since both are marked), and to pairs $$(\iv{\sigma_1'\sigma_1^{-1}\sigma\mu_a},\iv{\sigma_2'\sigma_2^{-1}\sigma\mu_a})\quad \text{ and} \quad (\iv{\sigma\mu_a},\iv{\sigma\mu_a})$$ to obtain
\[g(u_1\sigma_1^{-1}\sigma\mu_a)=h(u_2\sigma_2^{-1}\sigma\mu_a),\]
	whence \eqref{16081} follows too.

 The rest is Lemma \ref{shortest1}.
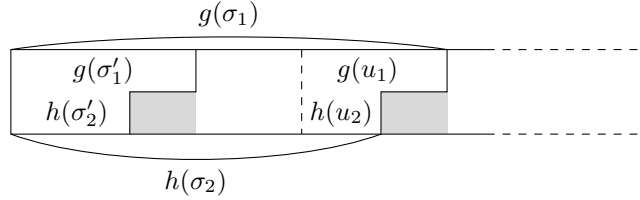
\[
\begin{tikzpicture}
\dolseda{15}{40}{} \dolseda{110}{135}{}
 \draw (150pt,-\radek)--(-30pt,-\radek)--(-30pt,\radek)--(150pt,\radek);
 \draw[dashed] (150pt,\radek)--(210pt,\radek); \draw[dashed] (150pt,-\radek)--(210pt,-\radek);
 \stredA{40}{15} \stredA{135}{110} 
 \horuzel{5}{$g(\sigma_1')$} \doluzel{-5}{$h(\sigma_2')$}
 \horuzel{105}{$g(u_1)$} \doluzel{95}{$h(u_2)$} \draw[dashed] (80pt,\radek)--(80pt,-\radek);
 \horuhornad{-30}{135}{$g(\sigma_1)$} 
 \doludolpod{-30}{110}{$h(\sigma_2)$} 
\end{tikzpicture}
\] 
\end{proof}

As a particular case, we point out the following corollary.

\begin{lem}\label{shortestcor}
	Let $(g,h)$ be a shortest counterexample such that $\iv g$ and $\iv h$ marked. Let two prefixes $\sigma_1$ and $\sigma_2$ of $\sigma$ satisfy $g(\sigma_1)=h(\sigma_2)z_h$. Then $\sigma_1=\sigma_2=\sigma$.
\end{lem}

%-----------------------------------------------------------------------------
\subsection{The case: $\p_1(\sigma)= a$ or $\s_1(\sigma)= a$}\label{sub1}
%----------------------------------------------------------------------------
\label{suba} 

In this subsection we show that the word $\sigma$ of a counterexample cannot start nor end by the letter $a$.

Since $\abs{g(a)}>\abs{h(a)}$ and $\s_1(\mu_c)=a$ for some $c\in A$, we have
\begin{align}\label{1022}
h(a)\in\s(g(a)).
\end{align}

\begin{claim}\label{spor pref a}
There is no counterexample such that both $\iv g$ and $\iv h$ are marked and $\p_1(\sigma)= a$ or $\s_1(\sigma)= a$.
\end{claim}
\begin{proof} 
Let first $\p_1(\sigma)= a$. As in the proof of Claim \ref{spor zg non-empty}, we obtain that $z_h$ and $h(a)$ have a common primitive root, say $t$. From \eqref{1022} we have that $t$ is a suffix of $g(a)$, which yields a contradiction with Claim \ref{nekomutuje s a}.

The case $\s_1(\sigma)= a$ follows from the same considerations for morphisms $\iv g$ and $\iv {\normh}$ by Lemma \ref{mirror}. 
\end{proof}

%-----------------------------------------------------------------------------
\subsection{The case: $\p_1(\sigma)=\s_1(\sigma)= b$}\label{subkonec}
%----------------------------------------------------------------------------
\label{subb}
%This is the final, most complicated part of the proof of Theorem \ref{hlavni veta}. 

In this subsection we shall suppose that $(g,h)$ is a counterexample such that $\iv g$ and $\iv h$ are marked, and $\p_1(\sigma)=\s_1(\sigma)=b$. We shall restrict possible counterexamples to the case $\mu_b\in b^+$.

We first fix some notation.

\begin{convention}\label{convention}
\ 
\begin{itemize}
%\item Denote by $\xi$ the word $\mu_a$ or $\mu_b$ such that $\s_1(\xi)=b$.  
\item Denote by  $\ell$ the maximal integer such that $b^{\ell}$ is a prefix of $\sigma$.
\item Denote by  $k$ the maximal integer such that $b^{k}$ is a prefix of $\mu_b\sigma$.
\item Denote by  $\ell'$ the maximal integer such that $b^{\ell'}$ is a suffix of $\sigma\mu_b$ or $\sigma\mu_a$ (the one of the two equality words ending with $b$).
\item Denote by  $k'$ the maximal integer such that $b^{k'}$ is a suffix of $\sigma$.
\end{itemize}
\end{convention}

We make use of Lemma \ref{mirror} and suppose that $k'\geq \ell$. In other words, we shall work either with $(g,h)$ or with $(\iv g, \iv{\normh})$ depending on whether $\sigma$ has more $b$s in the front or in the rear.

We first present some auxiliary lemmas.

\begin{lem}\label{hb gb nekomutuji}
The words $g(b)$ and $h(b)$ do not commute.
\end{lem}
\begin{proof}
Suppose, for a contradiction, that $t$ is the common primitive root of $g(b)$ and $h(b)$. Since $\abs{g(b)}<\abs{h(b)}$, we deduce, by $g(\sigma)=h(\sigma)z_h$, that the first occurrence of $g(a)$ in $g(\sigma)$ is comparable with $t$, a contradiction with $g$ being marked.
\end{proof}

\begin{lem}\label{g(b^l)<h(b)}
$\abs{h(b)}>(\ell+\ell'-1)\abs{g(b)}$.
\end{lem}
\begin{proof} The word $h(b)$ is comparable with $g(b)^{\ell}$ and suffix-comparable with $g(b)^{\ell'}$. If $\abs{h(b)}\leq(\ell+\ell'-1)\abs{g(b)}$, then $g(b)$ and $h(b)$ commute by Lemma \ref{factor}\eqref{factor3}, a contradiction with Lemma \ref{hb gb nekomutuji}.
\medskip
\[
\begin{tikzpicture}
 \horA{0}{30}{$g(b)$} \horA{30}{60}{$g(b)$} \horA{60}{90}{$g(b)$} \horA{90}{120}{$g(b)$} 
 \hordashA{120}{170}{}
 \dolA{80}{110}{$g(b)$} \dolA{110}{140}{$g(b)$} \dolA{140}{170}{$g(b)$}
 \doldashA{0}{80}{}
 \doludolbra{0}{170}{$h(b)$}
\end{tikzpicture}
\]
\end{proof}

\begin{lem}\label{dlouhe zh}
$\abs{z_h}>(\ell+k'-1)\abs{g(b)}$.
\end{lem}
\begin{proof}
The word $z_h$ is comparable with $g(b)^{\ell}$, since $z_h$ is comparable with $h(b)$, and $g(b)^{\ell}$ is a prefix of $h(b)$. Also $z_h$ is suffix-comparable with $g(b)^{k'}$, by $g(\sigma)=h(\sigma)z_h$.

First suppose that $\abs{z_h}\geq \abs{g(b)}$. Now, if $\abs{z_h}\leq (\ell+k'-1)\abs{g(b)}$, then $z_h$ and $g(b)$ commute by Lemma \ref{factor}\eqref{factor3}, a contradiction with Claim \ref{nekomutuje s b}.

Suppose now that $z_h$ is shorter than $g(b)$. Recall that $g(b)$ is a prefix of $g(\mu_b)$, prefix of $h(b)$, and a suffix of $g(\sigma)$. From $g(\sigma)=h(\sigma)z_h$ and $z_hg(\mu_b)=h(\mu_b)$ we deduce that there is a word $v$ such that $g(b)=vz_h$ and at the same time $g(b)=z_hv$. Again, the words $g(b)$ and $z_h$ commute, a contradiction. 
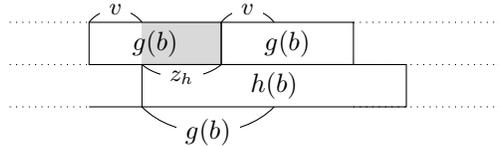
\[
\begin{tikzpicture}
\dolunulabra{20}{50}{{\small $z_h$}}
\horseda{20}{50}{}
\horuhorbra{0}{20}{{\small $v$}}
\horuhorbra{50}{70}{{\small $v$}}
\tecky{-30}{0} \draw (0,-\radek)--(20pt,-\radek);
\tecky{100}{160}
\horA{0}{50}{$g(b)$}  \horA{50}{100}{$g(b)$}
\dolA{20}{120}{$h(b)$}
\doludolbra{20}{70}{$g(b)$}
\end{tikzpicture}
\]
\end{proof}

An important step in the proof is the following lemma which shows that $g(a)$ cannot be too short.

\begin{lem}\label{existujeu}
$|g(ba)|>|h(b)|$.
\end{lem}

\begin{proof} 
In this proof we shall consider occurrences of $g(b)$s and $h(b)$s in $g(\sigma)$ and $h(\sigma)$, and their relative position. The idea is quite intuitive, but we give a more formal definition. Let $i,j\leq \abs{\sigma}_b$ be positive integers. Denote by $u_i$ the prefix of $\sigma$ such that also $u_ib$ is a prefix of $\sigma$, and $\abs{u_ib}_b=i$.

We say that the $i$th occurrence of $g(b)$ in $g(\sigma)$ {\em starts} within the $j$th occurrence of $h(b)$ in $h(\sigma)$, if  $$\abs{h(u_j)}\leq \abs{g(u_i)}< \abs{h(u_jb)}.$$ 
Similarly, we say that the $i$th occurrence of $g(b)$ in $g(\sigma)$ {\em ends} within the $j$th occurrence of $h(b)$ in $h(\sigma)$, if $$\abs{h(u_j)}< \abs{g(u_ib)}\leq \abs{h(u_jb)}.$$

Lemma \ref{dlouhe zh} immplies that the last occurrence of $g(b)$ in $g(\sigma)$ both starts and ends outside $h(\sigma)$. Therefore, by the pigeon hole principle, there is an occurrence of $h(b)$ in $h(\sigma)$ such that no occurrence of $g(b)$ in $g(\sigma)$ starts within it.   Similarly, there is an occurrence of  $h(b)$ in $h(\sigma)$ within which no $g(b)$ ends. 
From this it is easy to deduce that $h(b)$ is a prefix of $sg(a)^+$, and a suffix of $g(a)^+p$
where $s$ is a suffix of $g(b)$ or $g(a)$, and $p$ is a prefix of $g(b)$ or $g(a)$. Let $t$ be the primitive root of $g(a)$.

By $g(\sigma)=h(\sigma)z_h$, the words $h(b)$ and $g(b^\ell a)$ are prefix comparable. Suppose that $g(b^\ell t)$ is a prefix of $h(b)$. From the fact that $g(b^\ell)g(a)$ is a prefix of $sg(a)^+$ we deduce by Lemma \ref{factor}\eqref{factor1} that $\iv g$ is not marked, a contradiction. Similarly, we obtain a contradiction with $g$ being marked, if $tg(b)^{\ell'}$ is a suffix of $h(b)$. Therefore  
\begin{align}\label{dveb}
\text{	$|h(b)|<|g(b)^\ell t|$ and $|h(b)|<|t g(b)^{\ell'} |$}
\end{align} 
and we are through if $\ell=1$.

Suppose $\ell\geq 2$. Again by a pigeon hole principle, there are at least two occurrences of $h(b)$ in $g(\sigma)$ with no starting $g(b)$. Therefore $h(b)$ is a prefix of $s_1g(a)^+$ and $s_2g(a)^+$, where $s_1$ and $s_2$ are proper suffixes of $g(b)$ or $g(a)$. Note that $s_1$ and $s_2$ are overflows in $\sigma$, whence $s_1\neq s_2$ by Lemma \ref{shortest2}. 
Suppose, for a contradiction, that $|g(ba)|\leq |h(b)|$. From $s_1\neq s_2$, it is then not difficult to deduce that $g(a)$ overlaps nontrivially with $g(a)^2$, whence it is not primitive and $|g(a)|\geq 2|t|$.
From this and from \eqref{dveb} we obtain
\[2|h(b)|<|g(b)^\ell t|+|g(b)^{\ell'} t|\leq(\ell+\ell')|g(b)|+|g(a)|\leq(\ell+\ell'-1)|g(b)|+|h(b)|,\]
a contradiction with Lemma \ref{g(b^l)<h(b)}.
 \end{proof}

We can now once more point out two commuting words.

\begin{lem}\label{hb komutuje}
The word $h(b)$ commutes with $z_hg(b)^{k-\ell}$.
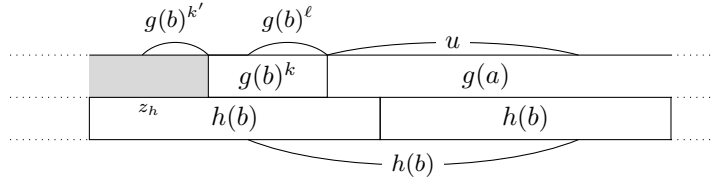
\[
\begin{tikzpicture}
\podpisA{0}{45}{{$z_h$}}
\horseda{0}{45}{}
%\horuhornad{0}{20}{{\small $s$}}
\horuhornad{20}{45}{{\small $g(b)^{k'}$}}
\horuhornad{60}{90}{{\small $g(b)^{\ell}$}}
%\horuhornad{60}{90}{{\small $g(b)^{k}$}}
\horuhorbra{90}{185}{$u$}
\doludolbra{60}{185}{{\small $h(b)$}}
\tecky{-30}{0} \draw (0,\radek)--(60pt,\radek);
\tecky{220}{240} \draw (90pt,\radek)--(220pt,\radek);
\horA{45}{90}{$g(b)^k$}
\dolA{0}{110}{$h(b)$} \dolA{110}{220}{$h(b)$}
\node at (150pt,0.5*\radek){$g(a)$};
\end{tikzpicture}
\]
\end{lem}
\begin{proof}
Lemma \ref{dlouhe zh} and the definition of $k'$ implies that $g(b)^{k'}$ is a suffix of $z_h$. The assumption $k'\geq \ell$ guarantees that $z_hg(b)^{k-\ell}$ is a well defined prefix of $z_hg(b)^k$.

 Let $u$ be the word $g(b)^{-\ell}h(b)$, which is a prefix of $g(a)$ by Lemma \ref{existujeu}. 
Since $\abs{h(b)}>\abs{g(b)}$ and $\ell\geq 1$, we have
$$\abs{h(b)^kz_h}>\abs{z_hg(b)^{k-\ell}h(b)}.$$ The equality $z_hg(\mu_b)=h(\mu_b)$ now implies that the word
\[z_hg(b)^{k}u=z_hg(b)^{k-\ell}h(b)\] is a prefix of $h(b)^+$ and thus $z_hg(b)^{k-\ell}$ commutes with $h(b)$ by Lemma \ref{factor}\eqref{factor3}.
\end{proof}

As a consequence, we have the claim of this section.

\begin{claim}\label{bb}
If $(g,h)$ is a shortest counterexample such that $\p_1(\sigma)=\s_1(\sigma)=b$, then $\mu_b= b^{k-\ell}$.
% and $k+\ell=k'+\ell'$.
\end{claim}
\begin{proof} 
Let $t$ be the common primitive root of words $h(b)$ and $z_hg(b)^{k-\ell}$. Recall that, by \eqref{vsudypritomnyprefix}, the maximal $t$-prefix of any $h(au)z_h$ is $z_h$. We deduce the following.
\begin{itemize}
	\item The maximal $t$-prefix of $h(\sigma)z_h=g(\sigma)$ is $h(b)^\ell z_h$. 
	\item The maximal $t$-prefix of $h(\mu_b\sigma)z_h=z_h g(\mu_b\sigma)=z_h g(b)^{k-\ell}g(b^{\ell-k}\mu_b\sigma)$ is $h(b)^kz_h$, which implies that the maximal $t$-prefix of $g(b^{\ell-k}\mu_b\sigma)$ is \[(z_hg(b)^{k-\ell})^{-1}h(b)^kz_h=h(b)^kg(b)^{\ell-k}.\]
 
\end{itemize}
 Since $\sigma$ contains $a$, both maximal $t$-prefixes mentioned above are proper.

Let first $h(b)^kg(b)^{\ell-k}\neq h(b)^\ell z_h$, and put $v=\sigma\wedge b^{\ell-k}\mu_b\sigma$.

 If $|h(b)^kg(b)^{\ell-k}|> |h(b)^\ell z_h|$, then $g(v)=h(b)^\ell z_h$, by Lemma \ref{gprefix}, a contradiction with Lemma \ref{shortestcor}. 

On the other hand, $|h(b)^kg(b)^{\ell-k}|< |h(b)^\ell z_h|$ implies $k<\ell$, and Lemma \ref{gprefix} yields $g(v)=h(b)^kg(b)^{\ell-k}$ and $g(vb^{k-\ell})=h(b)^k$, a contradiction with Lemma \ref{shortest1}.

It remains that $h(b)^kg(b)^{\ell-k}= h(b)^\ell z_h$, which implies $h(b)^{k-\ell}=z_hg(b)^{k-\ell}$ and $\mu_b=b^{k-\ell}$.
\end{proof}

%-----------------------------------------------------------------------------
\subsection{The case: $\p_1(\sigma)=\s_1(\sigma)= b$ and $\mu_b=b^{k-\ell}$}\label{mub}
%----------------------------------------------------------------------------

This last case is most difficult because it in a way compresses two places we use for the analysis into one, namely the beginning of $\sigma$ and the beginning of $\mu_b$. 
\[
\begin{tikzpicture}
\preghA{10}
\ramec{0}{250}
\stredA{80}{60} \stredA{190}{170}
\uzel{35}{$\sigma$} \uzel{95}{$b^{k-\ell}$} \uzel{125}{$b^{\ell}$} \uzel{150}{$( b^{-\ell}\sigma)$} \uzel{220}{$\mu_a$}
%\podpisA{40}{70}{$z_h$} \podpisA{130}{150}{$\kk z_h$}
\draw[dashed](110pt,\radek)--(110pt,-\radek);
\end{tikzpicture}
\]

Therefore, we have to employ a more detailed analysis of $\mu_a$.
\begin{claim}
	There is no shortest counterexample with $\p_1(\sigma)=\s_1(\sigma)=b$.
\end{claim}
\begin{proof}
Claim \ref{bb} implies $\p_1(\mu_b)=\s_1(\mu_b)=b$ whence \[\p_1(\mu_a)=\s_1(\mu_a) = a.\] Let $v_1$ be the longest prefix of $\mu_a$ ending with $a$ and satisfying $|z_hg(v_1)|>h(v_1)$. It follows that $\mu_a=v_1b^mv_2$ where $m>0$, $\s_1(v_1)=\p_1(v_2)=a$ and
\begin{align*}
	|z_hg(v_1)|&>|h(v_1)|, &  |g(v_2)|&>|h(v_2)|.
\end{align*}
Denote $u_1=h(v_1)^{-1}z_hg(v_1)$ and $u_2=g(v_2)h(v_2)^{-1}$.
\[
\begin{tikzpicture}
\horuhorbra{50}{90}{$u_1$}
\horuhorbra{140}{200}{$u_2$}
\horA{90}{140}{$g(b)^m$}
%\horA{10}{90}{$g(a)$}
%\horA{160}{240}{$g(a)$}
\dolA{50}{200}{$h(b)^m$}
\tecky{0}{50} \tecky{200}{250}
\draw (50pt,\radek)--(200pt,\radek);
\horuzel{50}{$g(v_1)$} \horuzel{200}{$g(v_2)$}
\doluzel{25}{$h(v_1)$} \doluzel{225}{$h(v_2)$}
%\horuhornad{140}{220}{$g(a)$}
\end{tikzpicture}
\]
From $|h(b)|>|g(b)|$ and $z_hg(b^{k-\ell})=h(b^{k-\ell})$ we obtain
\begin{align}\label{kratkehb}
	|h(b)|\leq|z_hg(b)|.
\end{align}

Let now $y$ be the prefix of $g(ba)$ of length $h(b)$ and let  $x$ be the word such that $xg(b)^{\ell'}=h(b)^{\ell'}$. From \eqref{kratkehb} we deduce that $u_1g(b)^{m-1}y$ is a prefix of $h(b)^+$. Also  $xg(b)^{\ell'+\ell-1}y$ is a prefix of $h(b)^+$. 
Lemma \ref{factor}\eqref{factor4} now implies that $u_1g(b)^{m-1}$ and $xg(b)^{\ell'+\ell-1}$ are suffix comparable.  Since $\iv g$ is marked and both $x$ and $u_1$ are suffix comparable with $g(a)$, we deduce $m=\ell+\ell'$. 

%Lemma \ref{factor}\eqref{factor4} further yields that $u_1x^{-1}$ or $xu_1^{-1}$  commutes with $h(b)$. Since $xg(b)^{\ell'}=h(b)^{\ell'}$, the equality $u_1g(b)^{\ell'}=(u_1x^{-1})xg(b)^{\ell'}$ implies that $u_1g(b)^{\ell'}$ commutes with $h(b)$ too.

From Lemma \ref{existujeu}, we obtain that $h(b)$ is a prefix of $g(b)^\ell g(a)$ whence the word $u_1g(b)^{\ell'}h(b)$ is a prefix of $h(b)^mz_h$. Lemma \ref{factor}\eqref{factor3} now implies that $u_1g(b)^{\ell'}$ commutes with $h(b)$.

Note that $u_1g(b)^{\ell'}$ is the maximal $h(b)$-suffix of $g(\sigma v_1b^{\ell'})$ and $xg(b)^{\ell'}$ is the maximal $h(b)$-suffix of $g(\sigma \mu_b)$.
Minimality of $\sigma\mu_a$ implies that 
\[u_1g(b^{\ell'})\neq h(b^{\ell'})=xg(b^{\ell'}).\]
 Let $v$ be the longest common suffix of $\sigma\mu_b$ and $\sigma v_1b^{\ell'}$. We apply Lemma \ref{gprefix} to $\iv g$ and obtain that $g(v)$, which is the longest common suffix  of $g(\sigma\mu_b)$ and $g(\sigma v_1b^{\ell'})$, is equal either to $u_1g(b)^{\ell'}$ or to $xg(b)^{\ell'}$. In both cases, $g(v)$ commutes with $h(b)$; let $t$ be their common primitive root. 

Since $\iv g$ is marked, the maximal $t$-suffix of $g(\sigma\mu_b)$ is $g(u)$ where $u$ is the maximal $v$-suffix of $\sigma\mu_b$. Since $\iv h$ is marked, the maximal $t$-suffix of $h(\sigma\mu_b)$ is $h(b^{\ell'})$. Therefore $g(\sigma\mu_b u^{-1} )=h(\sigma \mu_b b^{-\ell'})=h(\sigma b^{-k'})$, where $\sigma\mu_bu^{-1}$ is a prefix of $\sigma$ since $h(\sigma b^{-k'})$ is a prefix of $h(\sigma)$. Hence $(g,h)$ is not a shortest counterexample by Lemma \ref{shortest1}.
\end{proof}

This concludes the proof that there is no counterexample. By Lemma \ref{uprava}, two minimal elements $\alpha$ and $\beta$ of $\E(g,h)$ cannot start with the same letter if $g$ and $h$ are both non-periodic. Clearly, also $\iv g$ and $\iv h$ are non-periodic and $\iv \alpha$, $\iv \beta$ are minimal elements of $\E(\iv g,\iv h)$. Theorem \ref{hlavni veta} is proved.

%#############################################################################################################
\section{Test set}\label{test}
%#############################################################################################################

In this section we show that each binary language has a test set of cardinality at most two. The result is a consequence of Theorem \ref{char} and Theorem \ref{hlavni veta}.

{\em Test set} of a language $L\subset \Sigma^*$ is defined as a subset $T$ of $L$ such that the agreement of two morphisms on the
language $T$ guarantees their agreement on $L$. Formally, for any two morphisms $g$ and $h$ defined on $\Sigma^*$
$$ (\forall\, u \in T)\ (\,g(u)=h(u)\,)\ \Rightarrow  (\forall\, v \in L)\ (\,g(v)=h(v)\,)\,.$$

The {\em ratio} of a word $u\in A^+$ is denoted by $r(u)$ and defined
by
$$r(u)=\frac{\ \abs u_a}{\ \abs u_b}\ . $$
If $\abs u_b=0$, then $r(u)=\infty.$ A word $u$ is said to be {\em ratio-primitive} if no proper prefix of $u$ has the
same ratio as $u$. 

It is not difficult to see that each nonempty word $u$ has a unique factorization 
$u=u_1\dots u_k$
where each $u_i$ is a nonempty ratio-primitive word such that $r(u_i)=r(u)$. We call it the {\em ratio-primitive factorization} of $u$. Let $\RP(L)$ denote the set of all ratio-primitive words $u$ such that $u$ occurs in the ratio-primitive factorization of at least one word in $L$.

\begin{lem}\label{ratio}
If $\abs{g(a)}\neq \abs{h(a)}$, then $\abs{g(u)}=\abs{h(u)}$, $u\in A^+$, if and only if 
\[r(u)=\frac{\abs{h(b)}-\abs{g(b)}}{\abs{g(a)}-\abs{h(a)}}.\]
If  $\abs{g(a)}= \abs{h(a)}$ and $\abs{g(b)}\neq \abs{h(b)}$, then $\abs{g(u)}=\abs{h(u)}$, $u\in A^+$, if and only if  
$r(u)=\infty$.
\end{lem}
\begin{proof}
Follows directly  from \[\abs{g(u)}=|u|_a\cdot\abs{g(a)}+|u|_b\cdot \abs{g(b)}\quad \text{and}\quad \abs{h(u)}=|u|_a\cdot\abs{h(a)}+|u|_b\cdot\abs{h(b)}.\] 
\end{proof}
An immediate corollary is the following fact.
\begin{lem}\label{staci rozklad}
Binary morphisms $g$ and $h$ agree on $L$ if and only if they agree on $\RP(L)$.
\end{lem}
Here is one more observation.
\begin{lem}\label{difratio}
If $g(u)=h(u)$ and $g(v)=h(v)$, with $u,v\in A^+$ and $r(u)\neq r(v)$, then $g=h$. 
\end{lem}
\begin{proof}
Since $r(u)\neq r(v)$, the word $uv$ contains both letters $a$ and $b$. Lemma \ref{ratio} implies $\abs{g(a)}= \abs{h(a)}$ and $\abs{g(b)}=\abs{h(b)}$ whence $g=h$.
\end{proof}
We can now proof the main claim.
\begin{thm}
Let $L\subset A^*$ be a language. Then $L$ possesses a test set of cardinality at most two.
\end{thm}

\begin{proof}
If $L$ contains words $u$ and $v$ with different ratios, then $T=\{u,v\}$ is a test set of $L$ by Lemma \ref{difratio}. 

Suppose that all words in $L$ have the same ratio.  
We first find a test set $T_\RP$ of cardinality at most two for $\RP(L)$. If $\RP(L)$ has cardinality at least three, let $T_\RP=\{u,v\}$ where $u,v\in \RP(L)$ and $\p_1(u)= \p_1(v)$. 

Let $g$ and $h$ be morphisms such that $g\neq h$, $g(u)=h(u)$ and $g(v)=h(v)$. Since $u$ and $v$ are ratio-primitive, Lemma \ref{ratio} implies that $u$ and $v$ are minimal elements of $\E(g,h)$. Therefore both morphisms are periodic by Theorem \ref{char}\eqref{char2} and Theorem \ref{hlavni veta}. By Theorem \ref{char}\eqref{char1}, we have $\RP(L)\subseteq \E(g,h)$.

Let now $T$ be a subset of $L$ such that $\RP(T)=T_\RP$. Clearly, $T$ can be chosen such that its cardinality is at most two. Lemma \ref{staci rozklad} concludes the proof.
\end{proof}

\input{revize.bbl}
\end{document}

%% file: obrazkyA.inc
\def\radek{16pt}
\newcommand{\horA}[3]
{
\draw (#1pt,0pt) rectangle (#2pt,\radek);
\node at ($(0.5*#2pt,0.5*\radek)+(0.5*#1pt,0)$){#3};
}
\newcommand{\dolA}[3]
{
\draw (#1pt,0pt) rectangle (#2pt,-\radek);
\node at ($(0.5*#2pt,-0.5*\radek)+(0.5*#1pt,0)$){#3};
}
\newcommand{\doldashA}[3]
{
\draw[dotted] (#1pt,0pt) rectangle (#2pt,-\radek);
\node at ($(0.5*#2pt,-0.5*\radek)+(0.5*#1pt,0)$){#3};
}
\newcommand{\hordashA}[3]
{
\draw[dotted] (#1pt,0pt) rectangle (#2pt,\radek);
\node at ($(0.5*#2pt,0.5*\radek)+(0.5*#1pt,0)$){#3};
}
\newcommand{\dolseda}[3]
{
\fill[black!15] (#1pt,0pt) rectangle (#2pt,-\radek);
\node at ($(0.5*#2pt,-0.5*\radek)+(0.5*#1pt,0)$){#3};
}
\newcommand{\horseda}[3]
{
\fill[black!15] (#1pt,0pt) rectangle (#2pt,\radek);
\node at ($(0.5*#2pt,0.5*\radek)+(0.5*#1pt,0)$){#3};
}
\newcommand{\doludolbra}[3]
{
\draw (#1pt,-\radek) .. controls ($(#1pt,-1.8*\radek)+0.25*(#2pt,0)-0.25*(#1pt,0)$) and ($(#2pt,-1.8*\radek)-0.25*(#2pt,0)+0.25*(#1pt,0)$) ..
node[fill=white] {\textcolor{black}{#3}}
(#2pt,-\radek);
}
\newcommand{\dolunulabra}[3]
{
\draw (#1pt,0) .. controls ($(#1pt,-0.4*\radek)+0.25*(#2pt,0)-0.25*(#1pt,0)$) and ($(#2pt,-0.4*\radek)-0.25*(#2pt,0)+0.25*(#1pt,0)$) ..
node[fill=white] {\textcolor{black}{#3}}
(#2pt,0);
}
\newcommand{\doludolpod}[3]
{
\draw (#1pt,-\radek) .. controls ($(#1pt,-1.8*\radek)+0.25*(#2pt,0)-0.25*(#1pt,0)$) and ($(#2pt,-1.8*\radek)-0.25*(#2pt,0)+0.25*(#1pt,0)$) ..
node[below] {\textcolor{black}{#3}}
(#2pt,-\radek);
}%%%%%%%%%%%%%%%%%%%%%%%%%%%%%%%%%%%%%%%%%%%%%%%%%%
\newcommand{\horuhorbra}[3]
{
\draw (#1pt,\radek) .. controls ($(#1pt,1.4*\radek)+0.25*(#2pt,0)-0.25*(#1pt,0)$) and ($(#2pt,1.4*\radek)-0.25*(#2pt,0)+0.25*(#1pt,0)$) ..
node[fill=white] {\textcolor{black}{#3}}
(#2pt,\radek);
}
\newcommand{\horuhornad}[3]
{
\draw (#1pt,\radek) .. controls ($(#1pt,1.4*\radek)+0.25*(#2pt,0)-0.25*(#1pt,0)$) and ($(#2pt,1.4*\radek)-0.25*(#2pt,0)+0.25*(#1pt,0)$) ..
node[above] {\textcolor{black}{#3}}
(#2pt,\radek);
}
\newcommand{\preghA}[1]
{
\node at (-#1pt,0.5*\radek) {$g:$};
\node at (-#1pt,-0.5*\radek) {$h:$};
}
\newcommand{\ramec}[2]
{
\draw (#1pt,-\radek) rectangle (#2pt,\radek);
}
\newcommand{\stredA}[2]%
{
\draw (#1pt,\radek)--(#1pt,0pt)--(#2pt,0pt)--(#2pt,-\radek);
}
\newcommand{\nadpisA}[3]
{
\path (#1pt,0pt)--(#2pt,0pt) node[above=-2pt, midway]{{\scriptsize #3}};
}
\newcommand{\podpisA}[3]
{
\path (#1pt,0pt)--(#2pt,0pt) node[below=-1pt, midway]{{\scriptsize #3}};
}%%%%%%%%%%%%%%%%%%%%%%%%%%%%%%%%%%%%%%%%%%%%%%%%%%
\newcommand{\uzel}[2]
{
\node at (#1pt,0pt) {#2};
}
\newcommand{\horuzel}[2]
{
\node at (#1pt,0.5*\radek) {#2};
}
\newcommand{\doluzel}[2]
{
\node at (#1pt,-0.5*\radek) {#2};
}
\newcommand{\jensigma}[2]
{
\draw (0,-\radek)--(0,\radek)--(#1pt,\radek)--(#1pt,0pt)--(#2pt,0pt)--(#2pt,-\radek) -- cycle;
}
\newcommand{\tecky}[2]
{
\draw[dotted] (#1pt,\radek)--(#2pt,\radek) (#1pt,0)--(#2pt,0) (#1pt,-\radek)--(#2pt,-\radek);
}

%% file: BinaryRevize2012.bbl
\begin{thebibliography}{10}

\bibitem{defaut}
J.~Berstel, D.~Perrin, J.-F. Perrot, and A.~Restivo.
\newblock Sur le th\'eor\`eme du d\'efaut.
\newblock {\em J. Algebra}, 60(1):169--180, 1979.

\bibitem{handbookcombinatorics}
C.~Choffrut and J.~Karhum{\"a}ki.
\newblock Combinatorics of words.
\newblock In G.~Rozenberg and A.~Salomaa, editors, {\em Handbook of formal
  languages, Vol.\ 1}, pages 329--438. Springer, Berlin, 1997.

\bibitem{CuKabinary}
K.~{\v{C}}ul{\'{\i}}k, II and J.~Karhum{\"a}ki.
\newblock On the equality sets for homomorphisms on free monoids with two
  generators.
\newblock {\em RAIRO Inform. Th\'eor.}, 14(4):349--369, 1980.

\bibitem{EKRequality}
A.~Ehrenfeucht, J.~Karhum{\"a}ki, and G.~Rozenberg.
\newblock On binary equality sets and a solution to the test set conjecture in
  the binary case.
\newblock {\em J. Algebra}, 85(1):76--85, 1983.


\bibitem{EKRPCP}
A.~Ehrenfeucht, J.~Karhum{\"a}ki, and G.~Rozenberg.
\newblock The (generalized) {P}ost correspondence problem with lists consisting
  of two words is decidable.
\newblock {\em Theoret. Comput. Sci.}, 21(2):119--144, 1982.

\bibitem{handbookmorphisms}
T.~Harju and J.~Karhum{\"a}ki.
\newblock Morphisms.
\newblock In G.~Rozenberg and A.~Salomaa, editors, {\em Handbook of formal
  languages, Vol.\ 1}, pages 439--510. Springer, Berlin, 1997.

\bibitem{unique}
{\v{S}}.~Holub.
\newblock A unique structure of two-generated binary equality sets.
\newblock In {\em Developments in language theory}, volume 2450 of {\em Lecture
  Notes in Comput. Sci.}, pages 245--257. Springer, Berlin, 2003.

\bibitem{oldlot}
M.~Lothaire.
\newblock {\em Combinatorics on words}, volume~17 of {\em Encyclopedia of
  Mathematics and its Applications}.
\newblock Addison-Wesley Publishing Co., Reading, Mass., 1983.

\bibitem{ArtoEquality}
A.~Salomaa.
\newblock Equality set for homomorphisms of free monoids.
\newblock {\em Acta Cybernet.}, 4(1):127--139, 1978/79.

\bibitem{myphd}
\v{S}t\v{e}p\'{a}n Holub.
\newblock {\em Equations in free monoids}.
\newblock {PhD} thesis, Charles University, Prague, 2000.

\end{thebibliography}
